\documentclass[11pt]{article}

\usepackage{natbib}
\usepackage[normalem]{ulem}
\newcounter{ex}
\newcommand{\examp}[1]{\refstepcounter{ex}\theex\label{#1}}
\usepackage[margin=1in]{geometry}

\usepackage{bbm}
\usepackage{caption}
\usepackage{subcaption}
\usepackage{xcolor}
\usepackage{tikz}
\usepackage{amsmath,amsthm,amssymb}
\usetikzlibrary{arrows}
\usetikzlibrary{decorations.markings}
\newcommand{\given}{\; | \;}
\renewcommand{\b}[1]{\left[ #1 \right]}
\newcommand{\p}[1]{\left( #1 \right)}
\newcommand{\E}[1]{\mathbb{E}\left[#1\right]}
\newcommand{\EE}[2]{\mathbb{E}_{#1}\left[#2\right]}
\newcommand{\ind}[1]{\mathbbm{1}_{#1}}

\newcommand{\mc}[1]{\mathcal{#1}}
\newcommand\numberthis{\addtocounter{equation}{1}\tag{\theequation}}
\def\vbar{\overline{v}}
\newtheorem{theorem}{Theorem}[section]

\newtheorem{lemma}[theorem]{Lemma}
\newtheorem{proposition}[theorem]{Proposition}

\DeclareMathOperator*{\argmax}{arg\,max}

\usepackage{hyperref}
\hypersetup{
    colorlinks=true,
    linkcolor=blue,
    filecolor=magenta,
    urlcolor=cyan,
    citecolor=blue,
}

\newcommand{\xhdr}[1]{\paragraph*{\bf #1.}}
\def\barv{\overline{v}}
\def\eps{\varepsilon}

\renewcommand{\b}[1]{\left[ #1 \right]}

\begin{document}
\title{The Challenge of Understanding What Users Want:
Inconsistent Preferences and Engagement Optimization}

\author{Jon Kleinberg\footnote{Departments of Computer Science
and Information Science, Cornell University.}
\and Sendhil Mullainathan\footnote{University of Chicago Booth School of
Business}
\and Manish Raghavan\footnote{MIT Sloan School of Management and Department of
Electrical Engineering and Computer Science}}
\date{}

\maketitle

\begin{abstract}

Online platforms have a wealth of data, run countless experiments and use
industrial-scale algorithms to optimize user experience. Despite this, many
users seem to regret the time they spend on these platforms. One possible
explanation is that incentives are misaligned: platforms are not optimizing for
user happiness. We suggest the problem runs deeper, transcending the specific
incentives of any particular platform, and instead stems from a mistaken
foundational assumption. To understand what users want, platforms look at what
users do.
This is a kind of revealed-preference assumption that is ubiquitous in the way
user models are built.
Yet research has demonstrated, and personal experience
affirms,  that we often make choices in the moment that are inconsistent with
what we actually want. The behavioral economics and  psychology literatures
suggest, for example, that we can choose mindlessly or that we can be too
myopic in our choices, behaviors that feel entirely familiar on online
platforms.

In this work, we develop a model of media consumption where users have
inconsistent preferences. We consider a
platform which wants to
maximize user utility, but only observes behavioral data in the form of the
user's engagement.
We show how our model of users' preference inconsistencies
produces phenomena that are familiar from everyday experience, but difficult to
capture in traditional user interaction models. These phenomena include users
who have long sessions on a platform but derive very little utility from it, and
platform changes that steadily raise user engagement before abruptly causing
users to go ``cold turkey'' and quit.  A key ingredient in our model is a
formulation for how platforms determine what to show users: they optimize over a
large set of potential content (the content manifold) parametrized by underlying
features of the content. Whether improving engagement improves user welfare
depends on the direction of movement in the content manifold: for certain
directions of change, increasing engagement makes users less happy, while in
other directions on the same manifold, increasing engagement makes users happier.
We provide a characterization of the structure of content manifolds for which
increasing engagement fails to increase user utility. By linking these effects
to abstractions of platform design choices, our model thus creates a theoretical
framework and vocabulary in which to explore interactions between design,
behavioral science, and social media.
\end{abstract}

\section{Introduction}
There is a pervasive sense that online platforms are failing to provide a
genuinely satisfying experience to users. For example, when people are
experimentally induced not to use social media, they do not appear less happy;
in fact, some of them appear
happier~\citep{tromholt2016facebook,allcott2021digital}. What
explanations do we have for these problems?  Generically, we can think of a
firm’s objective as a combination of joint
surplus (make users happier) and surplus extraction (transfer from users to
firms). So, when platforms fail to provide satisfying experiences, a natural
explanation is that platforms are extracting surplus (e.g., optimizing ad
revenue to the detriment of users). Here we show, though, that there is a deeper
problem, one that transcends the narrow objective functions of any given
platform and arises even when platforms are simply trying to make users
happy.

The problem is a mismatch between our intuitive understanding of people and our
formal models of users.  User models have not just been immensely useful for
online platforms; they have been essential. Platforms loaded with a wealth of
content and numerous design choices must repeatedly answer the question, ``What
do users want?'' User modeling allows them to convert the wealth of data on user
behavior---click rates, dwell time, engagement---into inferences about user
preferences~\citep{aggarwal2016recommender}.
Whether for algorithmic curation or interpreting the results of A/B testing,
these inferences from data underlie most online platforms.
Despite their apparent success, though, most such models are built on a faulty
assumption.

Consider the following example. You are at a party and there is a bowl of potato
chips in front of you, and before you know it you have eaten nearly all of the
chips.
Would your host be right to conclude that you really enjoy chips and they should
refill the bowl? Possibly. Perhaps you really loved them and want more. Or
perhaps you want the host to save you from yourself, and you are hoping that
they will put the bowl elsewhere even as they are in the process of refilling
it. Moreover, the answer might well be food-dependent; perhaps you would want
the host to refill a bowl of salad, but are hoping they don't refill the bowl of
chips.
This example and
others like it have been empirically analyzed for several decades
~\citep{thaler2015misbehaving}. Two salient points emerge. First, your behavior
(whether to eat another chip) does not match your preferences and so you can end
up eating more chips than you would like. Second, the host who mistakes your
behavior for your preferences might make matters worse by giving you yet more
chips.

The models of users implicit in textbooks, the ones that underlie platform
design and optimization, ignore the possibility of such choice inconsistencies \citep{ekstrand2016behaviorism,lyngs2018so}. 
Instead, platforms optimize for behavioral signals such as ``likes, comments and
playtime'' when ranking content~\citep{nyt-tiktok}.
These choice inconsistencies appear for many
reasons. For example, people might over-weight the present: we are myopic now
but want our future ``selves'' to be patient. Such inconsistencies signal a
deeper failure: a violation of revealed preference, which
traditional user models
implicitly assume to go from measured behavior to inferred
preference~\citep{aggarwal2016recommender}.
Abstracting from the specific psychological
mechanisms~\citep{thaler1981economic,akerlof1991procrastination,laibson1997golden,o1999doing,fudenberg2006dual},
a straightforward way to understand inconsistencies is
formalized in the ``two minds'' approach: one ``self'' is
impulsive and myopic while the other ``self'' is forward-looking and
thoughtful.
We denote these two ``selves'' as  ``system 1'' and ``system 2'' for
evocativeness,
though that terminology has additional psychological connotations we do not rely
on here.
Crudely put, in our language, system 2 will have a target number of chips it wants
to eat (or perhaps none at all), while system 1 will simply get drawn by the
next chip and eat it, irrespective of how
many have been consumed.
This example
illustrates a category of findings that call into question revealed
preference: people can choose things they don't want, and as a result 
choices do not always reveal preferences. 
One might say a person has two inconsistent sets of preferences, but here we
take the perspective of the long-run self (system 2), viewing
system 2's preferences as the ``actual'' preferences of the individual,
and the impulsive system 1 preferences as impeding the realization
of these actual preferences.
What happens when we apply this perspective to the problem of
modeling user preferences and designing social media
platforms based on such models?

\xhdr{The present work: Platform design when users have inconsistent
preferences}

In this paper, we explore the consequences of platforms' user modeling in the
case where users have conflicts within themselves. This issue has begun
to receive attention in several recent
papers~\citep{lyngs2019self,lyngs2020just,milli2021optimizing,allcott2021digital},
largely through empirical investigation; here, we contribute to this growing
line of work by proposing a theoretical foundation for preference
inconsistencies and digital platforms. In our model, we assume the platform is
optimizing for user welfare, so our goal is not to understand the distortions
created by profit maximization but instead those created by the user model
itself.\footnote{We view this approach is essential even if one believes profit
  maximization is a large part of the problem. Even in that case, experience
  with other domains has shown that the interaction of profit maximization with
  rich user (or consumer) psychology creates inefficiencies in complex ways that
  cannot be understood without a clear model of user behavior. For example, one
  of the key lessons from efforts at regulation in such domains is that these
  types of models are crucial for delineating the types of practices to be
regulated~\citep{alm2017using}.}

Since the underlying psychology is rich, there is no single standard
model. Instead, we build a simple model of such conflicts to
understand how they play out in an online platform context, drawing on models
from the economics and computer science
literatures~\citep{o1999doing,laibson1997golden,kleinberg2014time}. 
Roughly speaking, our basic model (described in more detail
beginning in the next section) supposes that a user 
encounters a stream of content on
a platform as a sequence of discrete items $t = 0, 1, 2, \ldots$.
For example, the items may be posts, tweets, or videos.
Each item produces a value for system 1 and a value for system 2.
When the user consumes item $t$, their system 1 response 
is impulsive and faster than their system 2 response; 
thus, if item $t$ appeals to system 1, then the user automatically moves on to
the next item
(thereby remaining on the platform) independent of whether system 2
would like to remain on the platform or not.
If item $t$ does not appeal to system 1, then control over 
the user's decision about whether to remain on the platform passes to system 2.

Unlike system 1, which simply reacts to the current item being consumed,
system 2 is forward-looking and decides
whether to remain on the platform
based on a prediction 
both of the future value of its own choices as well as the value
(positive or negative) of the items consumed because of system 1's choices.
In the language of behavioral economics, system 2 is sophisticated (as opposed
to system 1, which is naive)~\citep{o1999doing}.
System 2 also controls the decision to go to the platform in the first
place, and this too is based on whether the expected value of a visit to the 
platform (due to its own choices and system 1's) is positive or negative.
We will see that this model provides a theoretically clean way to
capture complex phenomena about the relation between users and platforms,
including the ways in which a user might consume more content than
they would like, and the ways in which a user might both avoid going to 
a platform but also spend long amounts of time on it when they do go to it.

In many ways, browsing and chip eating have much in common. Because we may
consume the next chip or piece of content even when we don’t want to, we can
end up feeling we over-consumed. The platform too can then be like a
well-intentioned host who, by giving us more of the chips or content that we
consume, actually makes matters worse. But the case of online content has
additional richness. The space of online content is immense and heterogeneous
and the platform has fine-grain control, in ways that are are qualitatively
larger than we see for most goods in the off-line world. Put crudely, it begins
to stretch the metaphor to imagine a party host who can bioengineer chips with
incredible precision to maximize your consumption, but that is what happens when
platforms choose from an ocean of content to maximize engagement.
We capture these effects in our model by representing each piece
of content in terms of underlying parameters that determine the user's
response to it;
this defines a space of possible content, and the platform's design
decisions are to select items from a feasible region
corresponding to what we term a \textit{content manifold} within this space.
We imagine a platform varying its content along this content manifold,
monitoring user behavioral metrics like session length (as a proxy for
engagement) as it does so.
In our analysis,  we will not focus on platform incentives---for
example, that platforms get paid through clicks. Instead, we will show
how a platform whose only objective is maximizing user well-being
might make poor choices because its models assume revealed
preference. Many of the problems associated with digital content, in
this model, can be explained through this mismatch between user
modeling and consumer psychology, without resort to mis-alignment
between user preferences and platform profits.

We show that when a platform increases engagement, this might correspond
to an increase in user welfare (as evaluated by system 2), but
it might also correspond to a reduction in welfare as system 1 comes
to control an increasing fraction of the decisions.
The difficulty in distinguishing between these scenarios lies at the
heart of several challenges for social media:
the challenge of designing for users with internal conflicts in 
their preferences, and the challenge in evaluating design decisions
from even detailed measurements and explicit A/B tests of user behavior.
We approach these challenges within our model through
a characterization of content manifolds
for which maximizing engagement does not maximize user utility, and we offer
suggestions for how one might tell different types of content manifolds apart.

\xhdr{Implications of our model}
Notice how a sophisticated host who recognizes these problems would
behave. First, they would understand that not all party foods are chips:
automatically refilling the salad bowl is perfectly fine.\footnote{Making
matters worse, the conflict between system 1 and system 2 works differently in
different people: some find nuts a problem, others do not.} Second, they would
recognize that the optimal solution is not to have no chips at all. These goods
do not necessarily rise to the level of
cigarettes or other addictive substances where banning them could be optimal.
The host, then, needs (i) a strategy for understanding how different kinds of
content may be more problematic and (ii) to use that knowledge to decide how
heavily to rely on consumption (engagement). Those are indeed the two challenges
platforms face in our model.

Our model seeks to address these challenges in a concrete, stylized form,
highlighting a number of crucial aspects.
First, it argues that a platform needs to accurately model
the internal conflicts 
in users' preferences when it makes decisions about its content.
It is wrong to ignore the addictiveness of online content (via its
appeal to system 1), but it is also wrong to assume that all online
content is addictive, or that all appeals to system 1 necessarily
reduce the welfare of users.
Rather, content is heterogeneous in these parameters and lies on
a content manifold of possible parameter values.
This leads to the second category of implications of our model, which we
explore in Section~\ref{sec:opt}:
that the manifold of feasible content and its structure
determine the extent to which metrics like engagement can serve---or fail
to serve---as reasonable proxies for measuring underlying user utility.
And when engagement and utility are misaligned, we show that this stems
from a small number of possible structural properties of the underlying content
manifold.
While the structure of this manifold is unobservable to the platform in our
model, in practice, platforms can take concrete steps to better understand the
features of their content, leading to a third category of implications:
platforms can apply behavioral models to learn the properties of their content.
Section~\ref{sec:strategies} illustrates how a platform might apply
behavioral models to leverage even small amounts of auxiliary data.
Finally, a fourth category of implications concerns the user
interface (UI) decisions that platforms make.
We argue that a broad range of UI designs, including
autoplay, embedded media, enforced breaks, and 
the ``width'' of a set of recommendations, all have natural
interpretations in our model of internal user conflict, 
and our model can therefore provide insight into the effect they may
have, consistent with prior
work~\citep{lyngs2019self,lyngs2020just,moser2020impulse}. Crucially, our model
also argues that UI design choices are not
separable from other content choices, since their effect depends
on where the platform has positioned its content on the underlying content
manifold.
Automatic refills from an attentive host will
differentially impact consumption of chips and salad; likewise, UI features like
autoplay will have different effects on engagement for different types of
content.
As a concrete illustration of how our model can provide insights
about design decisions, we show in Section~\ref{sec:tree} how to use it for reasoning about
{\em how many} choices to offer a user when presenting content recommendations.

More broadly, our model attempts to grapple concretely with 
the often-expressed sense that something is broken with online content
platforms. It sometimes becomes very tangible when a user announces
publicly that they will be deleting their account. This same user often
is the one who has the highest engagement. Such ``cold turkey'' behavior
begs comparison to other addictive commodities. Our model formalizes
the logic behind such comparisons, but also shows how the heterogeneity
of the content and the design flexibility of the platform 
lead to a richer picture in reality.
Social media can be like potato chips
or it can be like salad. The outcome depends, of course, on the
underlying content. It also crucially depends on the choices of
platforms. Currently, these choices are guided by a misleading model of
people's preferences.  That means even well-intentioned platforms can
end up serving chips when they think they are serving salad.

\section{A Model of User Engagement and Utility}
\label{sec:model}

We now describe the basic version of our model, which begins
as in the previous section with 
a user who encounters a stream of content (posts, tweets, videos)
on a platform as a sequence of discrete items $t = 0, 1, 2, \ldots$.

\subsection{User Decisions about Consumption and Participation}

\xhdr{System 2}
The user is a composite of two distinct decision-making agents,
whom we term {\em system 1} and {\em system 2}.
System 2 experiences the sequence of item as follows:
\begin{itemize}
\item 
Item $t$ produces utility $I_t \cdot v_t$ for system 2 if it is consumed, where
$v_t$ represents the value of the content to the user and $I_t$ captures
diminishing returns over session length.
We will assume that $v_t$ are independent and identically 
drawn from a distribution $\mc V$, 
and that the user knows the mean of this distribution $\barv$.
$\{I_t\}_{t=0}^\infty$ is a sequence of random variables
variable that capture diminishing returns. We consider different choices of
$I_t$ in what follows.
\item 
If the user stops their session and leaves the platform 
before consuming item $t$, they would receive a utility of $W$ 
from their outside activities away from the platform;
thus, $W$ is the opportunity cost of consuming the next item,
and system 2's net utility from item $t$ is $I_t v_t - W$. Without loss of
generality, we will assume for much of the paper that $W = 1$. (Equivalently, we
could scale $\barv$ by $1/W$.)
\item 
In our basic version of the model, we instantiate $I_t$
by positing that for some $q > 0$,
there is a probability $q$ after each item that system 2 wants
to continue, and a complementary probability $1 - q$ that system 2
views itself as ``done'' and derives no further utility.
In effect, this determines a randomly distributed target session length
for system 2, after which it stops accumulating utility.
(In Appendix~\ref{app:equiv} we describe equivalent constructions that
lead to this formulation.) Formally, $I_t$ is the waiting time for a Bernoulli
event with bias $1-q$. While the bulk of the paper will consider this
instantiation of $I_t$, in Section~\ref{sec:gamma}, we present an alternative
formulation with slightly different properties.
\end{itemize}

If system 1 played no role in the user's decisions, then 
each of the user's decisions would reflect system 2's preferences,
step by step.
In particular, without system 1,
the user would participate in the platform if and only if 
$\barv$, the expected value of $v_t$, is at least $W$.  
After each step on the platform, the user would decide they
are done with probability $1 - q$;
this means that the length of the user's session is a random variable $X$
distributed geometrically,
as the waiting time for a Bernoulli event of probability $1 - q$.
Thus the user would achieve an expected net utility of
$\barv - W$ over each of the steps $t$ from $1$ to $X$.
From an increase in session length we would infer there had been an increase in
q, and therefore a corresponding increase in the user’s utility.
In this sense, the
revealed-preference assumption is a reasonable one in the absence of system 1:
longer sessions mean happier users.\footnote{As we will see in
  Section~\ref{sec:opt}, there are additional nuances here. Engagement and utility
can be misaligned even in the absence of system 1.}
In keeping with the intuition flowing from this type of scenario,
we will refer to session length as {\em engagement},
and when platforms use the revealed-preference assumption to 
try maximizing engagement, we can think of this as an apparently reasonable
approach to maximizing user utility when system 1 plays no role in a user's
decisions.

\xhdr{System 1}
The situation becomes more subtle when system 1 also plays a role in 
the decisions about consuming content.
\begin{itemize}
\item 
We assume that each item $t$ also produces utility for system 1,
in this case a value $u_t$ drawn from a distribution $\mc U$.
The items $u_t$ are independent of each other, but $u_t$ may be
correlated with $v_t$ (we will consider this in more detail below).
Let $p$ be the probability that $u_t > 0$.
\item 
When the user consumes item $t$, their system 1 response 
is impulsive and faster than their system 2 response; and so if $u_t > 0$ then
the user automatically consumes the next item $t+1$ without considering
leaving the platform.
\item 
If $u_t \leq 0$, then system 1 takes no action, and control of
the user's decision about whether to remain on the platform passes to system 2.
\end{itemize}
Thus it suffices to assume that there is some fixed probability $p$ with
which system 1 is active at any given timestep. In general, this may not be
true; content at the top of a ranking my differ distributionally from content
further down. However, if users' sessions are short relative the pool of
eligible content that a platform maintains (despite being subjectively too long
from the user's perspective), then the distribution from which content are drawn
will not change significantly over the course of a session.

System 2 is forward-looking, however, and so when system 1 plays
a role, system 2's decision whether to remain on the platform becomes more complex:
it needs to evaluate its expected utility over all future steps,
including the steps $j$ in which system 1 makes the choice regardless
of whether system 2 is done consuming content or not.
Similarly, system 2 must engage in this reasoning when deciding
whether to participate in the platform in the first place.
There are three key parameters that govern this choice:
(i) $\barv/W$, which determines the relative utility system 2 experiences
per step;
(ii) $q$, which determines how long system 2 continues to derive utility
from items in a session;
and (iii) $p$, which determines the probability that system 1 controls
any particular decision about consumption.
We will think of $\barv$ as determining the {\em value} of the
content, $q$ as determining the {\em span} of system 2's
interest in the content, and $p$ as determining the
{\em moreishness} of the content---the extent to which system 1
wants to consume ``one more item,'' analogously to someone wanting
to consume one more potato chip.
Moreishness is clearly related to our earlier discussion of
addiction, but it doesn't include all of the behavioral properties of addiction.
Formal models of addictive substances don't just model
preference inconsistency; they also model how past consumption changes
physiology and therefore future utility, both with and without
consumption~\citep{becker1988theory}.

In order to implement this behavioral
model, a user need not know all of their parameters. It suffices for the user to
be able to accurately estimate their expected utility from using the
platform.
For simplicity, our model assumes that content are (1) independent and identically
distributed and (2) not dynamically updated. Removing these assumptions is an
interesting direction for future work.

\xhdr{The user's utility and participation decisions}
The model provides a clean way to think about how these key
parameters govern system 2's decision
whether to participate in the platform.
If system 2 participates, then 
it will derive expected net utility $\barv - W$ per step for a random number of
steps $X$ until system 2 decides it is done (with probability $q$ each
step).\footnote{Note that this is true even if $u_t$ and $v_t$ are correlated.}
Once it decides it is done, however, the user cannot necessarily
leave the platform immediately; rather, 
system 2 is still at the mercy of 
system 1's item-by-item decisions to continue.
As long as $u_t > 0$, which happens with probability $p$,
system 1 will continue even though system 2 is
now deriving net utility $-W$ per step.

This means that for purposes of analysis we can think of the user's
session as divided into two phases.
\begin{itemize}
\item The first phase runs up until the point at which system 2 decides
it is done and stops collecting further utility.
As noted above, the length of this session is distributed as
the waiting time for a Bernoulli event of probability $1 - q$.
It therefore has expected length $1 / (1-q)$, and since system 2
collects expected utility $\barv - W$ in each of these steps,
system 2's expected utility in the first phase is
$\dfrac{\barv - W}{1 - q}$.
\item The second phase 
is the remainder of the session after the first phase ends.
The second phase only lasts as long as system 1 remains engaged, independently
with probability $p > 0$ for each item.
It is therefore the waiting time for a Bernoulli event of probability 
$1 - p$, but with the additional point that this probability is applied
to the first step as well: if system 1 isn't interested in the last item
that system 2 wanted to consume in the first phase, then system 2 can
immediately leave, and this second
phase has a length of zero.
The expected length of the second phase is therefore
$p / (1-p)$;
system 2 collects utility $-W$ in each of these steps
(since it derives no utility from the content, and loses $W$ from
the opportunity cost of remaining on the platform), and 
therefore system 2's expected utility in the second phase is
$-\dfrac{pW}{1-p}$.
\end{itemize}

Adding up these expectations over the two phases, we see that
System 2's expected utility is a random variable $S$ with expected value
\begin{equation}
  \E{S} = \max\p{\frac{\barv - W}{1 - q} - \frac{pW}{1 - p}, ~ 0}.
\label{eq:utility}
\end{equation}
The user doesn't visit the platform at all when $\frac{\barv - W}{1 - q} -
\frac{pW}{1 - p} < 0$, since doing so would result in negative utility.
The user's session length is a random variable $T$ that is deterministically
equal to
$0$ when the user doesn't visit the platform, and otherwise has expected value
\begin{equation}
\E{T} = \frac{1}{1 - q} + \frac{p}{1 - p}.
\label{eq:engagement}
\end{equation}
When the parameters are not clear from context, we will sometimes write these as
$\E{S(p, q, \barv)}$ and $\E{T(p, q, \barv)}$.

Equations (\ref{eq:utility}) and (\ref{eq:engagement}) 
make clear how the conflict within the
user plays out in expected utility, and it also suggests some of the 
challenges we'll see in interpreting a user's behavior when the
user has internal conflict.
First, when $p = 0$, we see that a user's engagement---as measured
by expected session length---grows monotonically in the user's
utility.
This motivates the use of engagement maximization as a heuristic for
improving user utility under the revealed-preference assumption that
system 1 plays no role in the user's behavior.
In contrast, when $p > 0$, engagement and utility are no longer
as closely aligned.
For example, when a user (via system 2) 
decides not to go to the platform at all,
we see from Equation (\ref{eq:utility}) that
this might be because $\barv < W$, rendering the first term negative;
but it might instead be the case that $\barv > W$ but $p$ is so large
that the negative second term outweighs the positive first term.
Similarly, we see from 
Equation (\ref{eq:engagement}) that
a longer session---greater engagement---could be the result of high span ($q$),
in which case system 2 is deriving high net utility,
or because of high moreishness ($p$), leading to a reduction 
in system 2's utility.
A crucial aspect of our argument is that 
without understanding the underlying parameters of the content it is
providing---$p$, $q$, and $\barv$ in the case of our model---the platform cannot
distinguish these possibilities purely from the user's
in-session behavior.

\subsubsection{An alternative formulation for diminishing returns.}
\label{sec:gamma}

A feature of our basic model of system 2 utility is that it is memoryless,
meaning that as long as system 2 is still interested in the platform ($I_t =
1$), the agent's decision to continue doesn't depend on the prior realizations
or the timestep $t$. This makes analysis particularly simple. Here, we present an
alternative, stateful model of system 2 utility and show that while the analysis
becomes more complex, the agent's behavior remains qualitatively similar.

Qualitatively, we still seek to encode the fact that the agent experiences
diminishing returns on the platform. Recall that in our prior model, system 2's
expected utility for consuming content at time $t$ is given by $\vbar \cdot
\E{I_t}$, where $I_t$ is the product of $t$ Bernoulli random variables with bias
$q$. Here, we instead consider a model where system 2 deterministically satiates
at a constant geometric rate, parameterized by some $\gamma \in [0, 1)$. In this
version of the model, $I_t = \gamma^t$, and system 2's expected utility at
timestep $t$ is given by $\gamma^t \vbar$. Note that system 2's participation
decisions are no longer independent of $t$: the longer the agent has been
consuming, the less utility it derives.

In Appendix~\ref{app:stateful}, we show that the agent's behavior is
qualitatively similar to our original model: there exists some $t^*$ (which is a
function of $(p, \gamma, \vbar/W))$ such that system 2 chooses to continue at any
$t \le t^*$ and, given the chance, leaves the platform at any $t > t^*$.
We then derive for the agent's expected consumption and utility:
\begin{align*}
  \E{T}
  &= \begin{cases}
    t^* + \frac{1}{1-p} & t^* \ge 0 \\
    0 & \text{otherwise}
  \end{cases} \\
  \E{S}
  &= \vbar \p{\frac{1-\gamma^{t^*}}{1-\gamma} + \gamma^{t^*} \frac{1}{1 - p
  \gamma}} - W \p{t^* + \frac{1}{1-p}}.
\end{align*}
In the remainder of this paper, we focus on our original model of $I_t$ as the
product of Bernoulli random variables. We defer further discussion and analysis
of the alternative formulation described here to Appendix~\ref{app:stateful}.

\subsection{Platform Decisions about Content: An Overview}

So far, our model has considered 
how the user's behavior can vary with
the underlying parameters. But this model of behavior only tells part of the
story: crucially, the parameters $p$, $q$, and $\barv$ describing the content
are not exogenous, but
determined by the platform itself. Online platforms extensively optimize user
experiences through techniques including A/B testing and machine learning. Our
goal here will be to use our model to analyze the impacts of platform
optimization on users.

In general, platform optimization is data-driven. Platforms typically collect
extensive behavioral data---clicks, watch time, etc.---over which they optimize
anything from user interface choices to the content they recommend to a user.
Even if the platform seeks to maximize the
utility of its users (in keeping with our initial
decision to focus on joint surplus instead of transfer of
surplus), in
practice, platforms have little or no data on the actual utility of their users.
As a result, the platform will follow the revealed preference assumption
discussed above: as we argued there, when a user's decisions are made only by
system 2, maximizing engagement corresponds to maximizing user utility.
This limitation applies to altruistic third parties as well: for example,
suppose a user could use a browser plugin to re-rank content on the platform
based on their activity (see, e.g., \citet{xu2008user}). Without further
information, the browser plugin, too, will be forced to simply prioritize
content that leads to higher engagement.

In the language of our model, the platform observes $T$ but not $S$, and when
making decisions, chooses the option that maximizes $\E{T}$. In practice, of
course, platforms have collected a variety of behavioral data far more granular
than session lengths; for simplicity, we focus on session length $T$ as a target
metric, and we consider the potential use of other forms of behavioral data in
Section~\ref{sec:strategies}.

Engagement and utility can diverge, and optimizing engagement may
fail to optimize utility. The extent of this divergence, however, depends
crucially on platform decisions and the underlying features of the content.
Under this intuition, platforms are best understood as analogous to food:
for some types of food (like junk food), unhealthiness is inseparable from
enjoyment; but for others (like salad), enjoyment and healthiness can be
positively correlated. Content (and food) may thus lie on different
metaphorical manifolds. Junk food occupies a manifold with low nutritional
value and high moreishness, while salad occupies a manifold with high
nutritional value and low moreishness. Within these manifolds, some junk food
may be healthier than others, and some salads may be more moreish than others,
but in general, junk food and salad have different properties. Optimizing for
engagement might therefore produce healthy salads but unhealthy candy.

But if content manifolds are inherently similar to manifolds for
different types of food, why do online platforms raise qualitatively new
questions compared to, say, a restaurant deciding its menu? There are two
important distinctions to draw in this analogy. First, online platforms and
restaurants differ in their power to observe consumer welfare. While restaurants
know the nutritional content of their products, online platforms typically don't
observe the utility users derive from their content. They instead rely on
proxies to estimate the impacts of their choices on user utility. Second, the
space in which online platforms operate is enormous: new types of content are
invented on a daily basis, in contrast to the relatively modest pace at which
new foods are brought to market. Together, these differences imply that
platforms must select content from complex, unknown manifolds with little or no
information on the true impact that these selections have on user well-being.

\section{Optimization on Content Manifolds: Basic Results}

Guided by this discussion,
we can think about a platform's optimization over
an underlying content manifold as follows.
Suppose the platform is currently serving content from a distribution with
some set of parameter values $p$, $q$, and $\barv$.
Now, suppose the platform wants to modify this distribution to increase 
the user's net utility.
In general, it is not reasonable to assume that a platform 
can directly measure or intervene on 
the underlying quantities $p$, $q$, and $\barv$; rather, it is more natural to
assume that it has
control over some features $x$ of the content it selects, and there is a latent
mapping from $x$ to these hidden values $(p, q, \barv)$.

Thus, a platform observes the features $x$ of its content and the 
session-level behavior of its users---i.e., their amount of engagement.
As it varies the content by selecting for different features $x$, it sees user
behavior change. What is actually taking place is that these modifications to
$x$ are causing changes to the underlying content parameters $(p,q,\barv)$:
varying $x$ is causing the platform to move around on a content manifold in
$(p,q,\barv)$-space, and inducing changes in user behavior as 
a result.

In the next section, we develop a general formalism for analyzing
this type of optimization, but it is useful to start with some
concrete examples that illustrate the basic phenomena that arise.

\subsection{Three illustrative examples}

\xhdr{Example \examp{ex:quality}: Increasing quality}
Suppose first that the platform starts in a state where its
content satisfies $p = 0$, $q = q_0 < 1$, and $\barv = v_0 > W$.
Thus,
the content produces no internal conflict in the user (since $p = 0$
and hence system 1 is never active), and the content is appealing enough 
that users will participate (since along with $p = 0$ we have 
$\barv > W$).

Suppose also that the platform has the ability to modify its content
along a manifold parametrized by a single value $z \in [0,\eps]$:
the point on the content manifold corresponding to $z$ has 
$$p = 0; \hspace*{0.15in} q = q_0 + z; \hspace*{0.15in} \barv = v_0 + z.$$
(We'll assume that $\eps$, the maximum value of $z$, is small enough
that $q = q_0 + z$ remains strictly below $1$, i.e., $\eps < 1-q_0$.)
As $z$ increases, both the value and the span of the content increase;
since both of these raise system 2's utility, we can think of this as
a change toward content of greater quality.
The user's engagement $\E{T}$ also increases in $z$.
This content manifold therefore illustrates the basic motivation for
engagement-maximization in a world of revealed preference:
by choosing the point on the content manifold that maximizes
$\E{T}$, the platform has also maximized the user's expected utility $\E{S}$.

\begin{figure}[t]
  \includegraphics[width=0.45\linewidth]{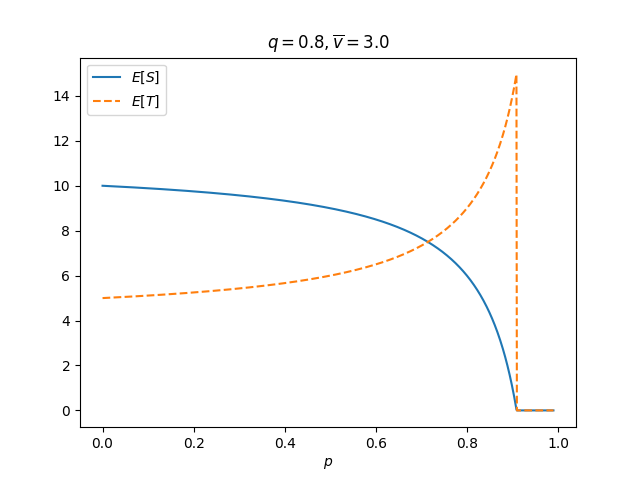}
  \includegraphics[width=0.45\linewidth]{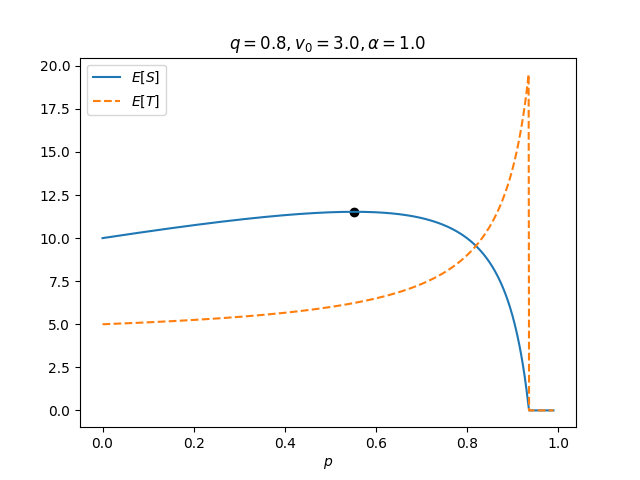}
  \caption{{\em Left panel (Example~\ref{ex:moreishness}):} For fixed $q$,
    $\barv$, and $W=1$, engagement and utility vary with $p$.  {\em Right panel
    (Example~\ref{ex:both}):} When $q$ and $W$ are fixed and $p$ and $\barv$ are
    positively correlated, engagement and utility are aligned up to a point.
    Here the expected value of each item ranges from $v_0$ when $p = 0$ to $v_0
  + \alpha$ when $p = 1$.}%
  \label{fig:consumption_surplus}
\end{figure}

\xhdr{Example \examp{ex:moreishness}: Increasing moreishness}
Suppose again that the platform starts in a state where its
content satisfies $p = 0$, $q = q_0 < 1$, and $\barv = v_0 > W$.
But now the content available to the platform is organized along
a different content manifold.
It is again parametrized by a single value $z$, but now
the point on the content manifold corresponding to $z$ has
$$p = z; \hspace*{0.15in} q = q_0; \hspace*{0.15in} \barv = v_0$$
for $z \in [0,1)$.
In other words, modifications to the content are entirely in the
direction of greater moreishness.

Now, as the platform increases $z$, engagement $\E{T}$
and utility $\E{S}$ vary as shown in 
the left panel of Figure \ref{fig:consumption_surplus}
(depicted with a specific choice of values for $q_0$ and $v_0$, though
the shape is qualitatively the same regardless of the exact
values for these quantities).
In particular, for small values of $z$, the engagement---which the 
platform can observe directly through session length---will go up, 
while the user's utility---which the platform cannot observe directly---will go
down.

Comparing Examples~\ref{ex:quality} and~\ref{ex:moreishness} for small values of
$z$ highlights a crucial fact that our model captures about platform
optimization: the 
platform cannot tell which case it is in purely from session-length
statistics.
In both cases, it is increasing $z$ and seeing engagement go up;
but this might be in service of greater user utility (when the content
lies on the content manifold in Example~\ref{ex:quality}) or it might be at the
expense of
user utility (when the content lies on the content manifold in
Example~\ref{ex:moreishness}).

As $z$ becomes sufficiently large, 
our model captures a second key phenomenon:
there comes a point in the trajectory toward larger $z$ when
the user's expected utility suddenly becomes negative, and
the user would stop participating in the platform:
at this point, the engagement $\E{T}$ drops discontinuously to $0$ 
with increasing $z$.
More precisely, it is not difficult to prove the following fact about
this example, which is also suggested
by the plot in the left panel of Figure \ref{fig:consumption_surplus}.

\begin{proposition}
  \label{prop:ex2}
The value of $z$ that maximizes $\E{T}$ produces
a utility of $0$ for the user.
\end{proposition}
\begin{proof}
  Let $z_S$ and $z_T$ be the values of $z$ that maximize $\E{S}$ and $\E{T}$
  respectively.
  First, note that because $p = z$, $\E{S}$ is trivially maximized by $z_S =
  0$, as $\E{S} = \frac{v-W}{1-q} - \frac{pW}{1-p}$.
  To show that $\E{S} = 0$ when $z = z_T$, define
  \begin{align*}
    z_T = p_T &\triangleq 1 - \frac{W}{\frac{v-W}{1-q} + W},
  \end{align*}
  and note that for $p = p_T$, $\E{S} = 0$. Observe that because $\E{T} =
  \frac{1}{1-q} + \frac{p}{1-p}$, whenever $\E{S} > 0$,
  \begin{align*}
    \frac{d}{dp} \E{T} = \frac{1}{(1-p)^2} > 0
  \end{align*}
  whenever $\E{S} > 0$. Thus, no $p$ with $\E{S} > 0$ can maximize $\E{T}$, and
  so $\E{T}$ must be maximized by $p$ such that $\E{S} = 0$. The only such $p$
  with non-zero $\E{T}$ (i.e., the agent chooses to use the platform at all) is
  $p_T = z_T$ as defined above.
\end{proof}

A platform that was using revealed preference to interpret session
length as a measure of engagement and hence utility
would thus be left with a puzzling set of observations to interpret in
this second scenario:
as it modified the content by varying $z$, engagement increased
steadily until the user abruptly went ``cold turkey'' and stopped
using the platform at all.

\xhdr{Example \examp{ex:both}: Interactions between moreishness and quality}
Examples~\ref{ex:quality} and~\ref{ex:moreishness} are extremely simple, in that
they allow the platform to modify some parameters while keeping others
purely constant.
Taken in isolation, they might create the superficial impression
that all increases in moreishness are bad, and all increases in
$q$ and $\barv$ are good.
But we would expect most scenarios to be more complex, 
involving situations where it is impossible to change one of
the parameters of the content without changing others,
and where it is therefore hard to make absolute statements about the
effect of any one parameter without taking into account
the overall structure of the content manifold.

A basic example of a content manifold in which changes to $p$
are correlated with changes to other parameters is one in which 
the content manifold is again
parametrized by a single value $z$, with
the point on the content manifold corresponding to $z$ having
$$p = z; \hspace*{0.15in} q = q_0; \hspace*{0.15in} \barv = v_0 + \alpha z$$
for $z \in [0,1)$ and constants $q_0 > 0$, $v_0 > W$, and $\alpha > 0$.
On this content manifold, modifications to content produce greater moreishness
but also greater expected value per item (with some potentially modest
slope $\alpha > 0$).
Such a content manifold is familiar from everyday experience, where 
content that engages system 1 might well also be more enjoyable for
system 2.

As $z$ increases, the engagement and utility vary as shown in
the right panel of Figure \ref{fig:consumption_surplus}:
for small values of $z$ the behavior is like Example~\ref{ex:quality}, with
engagement and utility growing together,
while beyond a certain point, the behavior is like Example~\ref{ex:moreishness},
with engagement growing while utility drops, until we get to
a value of $z$ large enough that the user chooses not to participate,
and engagement discontinuously drops to 0.
Formally, we can show the following fact about this example:
\begin{proposition}
  \label{prop:ex3}
The value of $z$
that maximizes utility is strictly between $0$ and $1$,
while the value of $z$ that maximizes engagement produces 0 utility
for the user.
\end{proposition}
\begin{proof}
  Let $z_S$ and $z_T$ be the values of $z$ that maximize $\E{S}$ and $\E{T}$
  respectively.
  In particular, we will show that
  \begin{align*}
    z_S
    &= \begin{cases}
      1-\sqrt{\frac{W(1-q)}{\alpha}} & \alpha \ge W(1-q) \\
      0 & \text{otherwise}
    \end{cases} \\
    z_T
    &= \frac{\p{\alpha - v_0 + Wq} + \sqrt{\p{\alpha - v_0 + Wq}^2 - 4 \alpha
    \p{W - v_0}}}{2 \alpha}.
  \end{align*}
  Assuming the agent uses the platform at all, their utility as a function of $z$
  is given by
  \begin{align*}
    \E{S}
    &= \frac{v_0 + \alpha z - W}{1-q} -
    \frac{zW}{1-z}.
  \end{align*}
  The derivative is
  \begin{align*}
    \frac{d}{dz} \frac{v_0 + \alpha z - W}{1-q} -
    \frac{zW}{1-z}
    &= \frac{\alpha}{1-q} - \frac{W}{(1-z)^2}.
  \end{align*}
  Note that the second derivative is negative, so $\E{S}$ is maximized when its
  derivative is 0 if such a point exists.
  \begin{align*}
    \frac{\alpha}{1-q} - \frac{W}{(1-z)^2}
    &= 0 \\
    \frac{\alpha}{1-q}
    &= \frac{W}{(1-z)^2} \\
    (1-z)^2 &= \frac{W(1-q)}{\alpha} \\
    z &= 1-\sqrt{\frac{W(1-q)}{\alpha}}
  \end{align*}
  Of course, this might be negative, in which case $\E{S}$ is maximized at $z =
  0$. Thus, the value of $z$ that maximizes utility is
  \begin{align*}
    z_S \triangleq \max\p{0, 1-\sqrt{\frac{W(1-q)}{\alpha}}}.
  \end{align*}

  As long as $\E{S} > 0$, consumption is monotonically increasing in $z$,
  meaning that as before, $z_T \triangleq \max \left\{z : \frac{v_0 + \alpha
  z-W}{1-q} - \frac{zW}{1-z} \ge 0\right\}$. Because $\frac{v_0 + \alpha
  z-W}{1-q} - \frac{zW}{1-z} < 0$
  for sufficiently large $z$, it suffices to find $z_T = \max \left\{z :
  \frac{v_0 + \alpha z-W}{1-q} - \frac{zW}{1-z} = 0\right\}$.
  \begin{align*}
    \frac{v_0 + \alpha z - W}{1-q} -
    \frac{Wz}{1-z}
    &= 0 \\
    z \p{\frac{\alpha}{1-q} - \frac{W}{1-z}}
    &= \frac{W - v_0}{1-q} \\
    z \b{\p{\frac{\alpha}{1-q}}(1-z) - W + \frac{W -
    v_0}{1-q}}
    &= \frac{W - v_0}{1-q} \\
    z^2 \frac{\alpha}{1-q} - z \p{\frac{\alpha}{1-q} - W +
    \frac{W - v_0}{1-q}} + \frac{W - v_0}{1-q} &= 0 \\
    z^2 \alpha - z \p{\alpha - W(1-q) + (W - v_0)} + (W - v_0) &= 0 \\
    z^2 \alpha - z \p{\alpha - v_0 + Wq} + (W - v_0) &= 0 \\
  \end{align*}
  Because $\E{T}$ is concave in $z$ and is strictly positive for $z = 0$, we
  want the larger of the two quadratic roots.
  \begin{align*}
    z_T
    &= \frac{\p{\alpha - v_0 + Wq} + \sqrt{\p{\alpha - v_0 + Wq}^2 - 4 \alpha
    \p{W - v_0}}}{2 \alpha}
  \end{align*}
  To show that $z_T > z_S$, it suffices to show that $\E{S} = 0$ for
  all $z \ge z_T$, since $\E{S} = 0$ for $z = z_T$. This holds because for $z
  \le z_T$, $\E{S}$ is concave over
  $z$ and $\E{S} > 0$ for $z = 0$.
\end{proof}

This example, despite its simple structure, illustrates
some of the complex phenomena that arise when platforms optimize
over an underlying content manifold.
It is too simplistic to say that increasing engagement is always
a good heuristic for increasing utility, although this is the case in 
Example~\ref{ex:quality}.
It is also too simplistic to say that increasing moreishness is
always bad, although this is the case in Example~\ref{ex:moreishness}.
Rather, engagement can be a good signal for utility
over some parts of the content manifold and a bad signal over other parts,
and the challenge for a platform is to understand the structure
of the content manifold well enough to know where these different effects apply.

\subsection{Mixtures of content sources}

One natural way to generate a content manifold is through a mixture of $k$
individual content sources, where each content source might be a particular
genre of content like sports highlights or science videos. A platform could
choose a user's feed as a weighted mixture between these sources: at time $t$, a
piece of content from source $i$ is chosen with probability $a_i$. Suppose each
content source has parameters $(p_i, q_i, \barv_i)$. Then, the following
proposition shows that this mixture yields the weighted average of the
parameters of each content source. Thus, given $k$ content sources, the platform
can achieve any convex combination of their parameters simply by mixing them
together.

\begin{proposition}
  \label{prop:mixture}
  If content with parameters $(p_i, q_i, \barv_i)$ is chosen with probability
  $a_i$ independently at each step $t$, then the resulting content distribution
  has parameters $\p{\sum_i a_i p_i, \sum_i a_i q_i, \sum_i a_i \barv_i}$.
\end{proposition}
\begin{proof}
  The key observation here is that the agent's behavior at time $t$ does not
  depend on the content randomly selected at time $t$, since it only observes
  this content after it makes any decisions at time $t$. The agent's belief
  about the future, and therefore its behavior, depends only on the expected
  characteristics of that content. At time $t$, the probability that system 1
  will be active at time $t+1$ is $\sum_i a_i p_i$. Recall that $I_t$
  denotes the event that System 2 is still deriving value from the platform at
  time $t$. Then, $\Pr[I_{t+1} = 1 \given I_t = 1] = \sum_i a_i q_i$.
  Finally, $\E{v_t} = \sum_i a_i \barv_i$. Therefore, the agent behaves as if
  the parameters are $\p{\sum_i a_i p_i, \sum_i a_i q_i, \sum_i a_i \barv_i}$.
\end{proof}

While our discussion so far has been framed in terms of choices over content, we
can also think of certain design decisions in the context of our model. For
example, some platforms allow users to insert break reminders into their feeds
as a way to manage their time on the platform. In our model, we can think of a
break as a piece of content with low moreishness, low value, and high span
(because it is unlikely to satisfy a user's true desire for content). We can use
Proposition~\ref{prop:mixture} to anticipate how inserting breaks into a content
feed might alter the overall distribution's characteristics in $(p, q, \barv)$
space: breaks will decrease both the moreishness and value of the content. If
the existing content already has low moreishness, this may have little impact on
the user's behavior and utility; but if existing content has high moreishness,
the addition of breaks might significantly reduce their time on the platform
while increasing their overall utility. We discuss this in further detail
in Section~\ref{sec:UI}.

\section{Optimization on Content Manifolds: Structural Characterization}
\label{sec:opt}
Here, we develop the model of content manifolds 
at a general level, abstracting from these specific examples
and their properties.
As part of this, we provide two characterization theorems
that link the structure of the content manifold to the outcome of
the platform's optimization.

Platforms often use techniques like A/B testing and machine learning to optimize
engagement as a function of features $x$ from some domain $\mc X$. This is a
canonical approach to content curation: a platform will extract features from
content and try to predict how a user will engage with that content. In our
model, each such $x$ is associated with parameters $(p, q, \barv)$, unknown to
the platform, that dictate behavior. The platform attempts to maximize
engagement, optimizing $\E{T}$ over $\mc X$.

More formally, the underlying parameters $(p, q, \barv)$ lie in the
3-dimensional space $\Omega \triangleq [0, 1) \times [0, 1) \times
\mathbb{R}^+$. If each $x \in \mc X$ corresponds to some $\omega \in \Omega$,
then $\mc X$ has a corresponding \textit{content manifold}, which we denote $\mc
M$, over $\Omega$. Formally, if we define $f_{\mc X}$ to be the mapping from
$\mc X$ to $\Omega$, then $\mc X$ induces the content manifold $\mc M \triangleq
\{\omega : \exists x \in \mc X \text{ s.t. } f_{\mc X}(x) = \omega\}$. We will
assume $\mc M$ is a closed set.

In what follows, we will take a content manifold $\mc M$ as given, though as a concrete
example, we can think of $\mc M$ as generated by a mixture of content sources as
in Proposition~\ref{prop:mixture}; we will return to the relationship between
$\mc X$ and $\mc M$ later. Examples~\ref{ex:moreishness} and~\ref{ex:both} from
the previous section can be written as the content manifolds $\mc M_2 = \{(p, q, \barv)
: p \in [0, 1)\}$ (for fixed $q$ and $\barv$) and $\mc M_3 = \{(p, q, \barv) : p
\in [0, 1), \barv= v_0 + \alpha p\}$ (for fixed $q$) respectively. To gain some
intuition for how utility and engagement behave, we can visualize both $\E{S}$
and $\E{T}$ over various content manifolds. For example,
Figures~\ref{fig:c_s_heatmap_fixed_v} and~\ref{fig:c_s_heatmpap_vp_vm_log}
depict $\E{S}$ and $\E{T}$ for the content manifolds $\mc M_2$ and $\mc M_3$
instantiated with different values of $q$.

\begin{figure}[ht]
  \centering
  \includegraphics[width=0.9\linewidth]{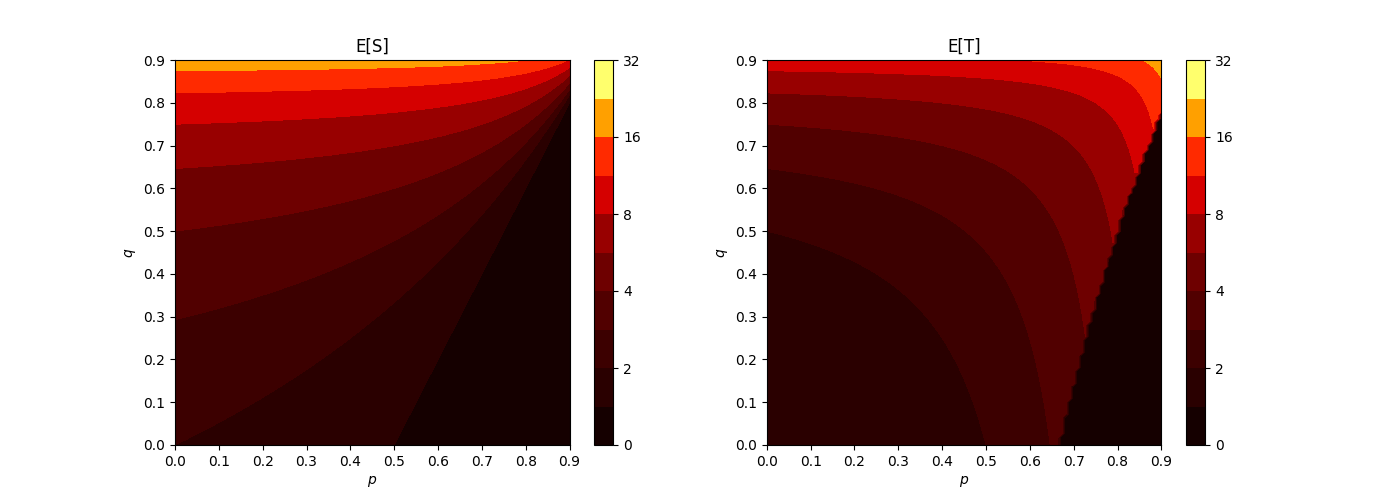}
  \caption{$\E{S}$ and $\E{T}$ plotted for the content manifold $\mc M_2$ from
  Example~\ref{ex:moreishness} with $\barv = 3$.}%
  \label{fig:c_s_heatmap_fixed_v}
\end{figure}

\begin{figure}[ht]
  \centering
  \includegraphics[width=0.9\linewidth]{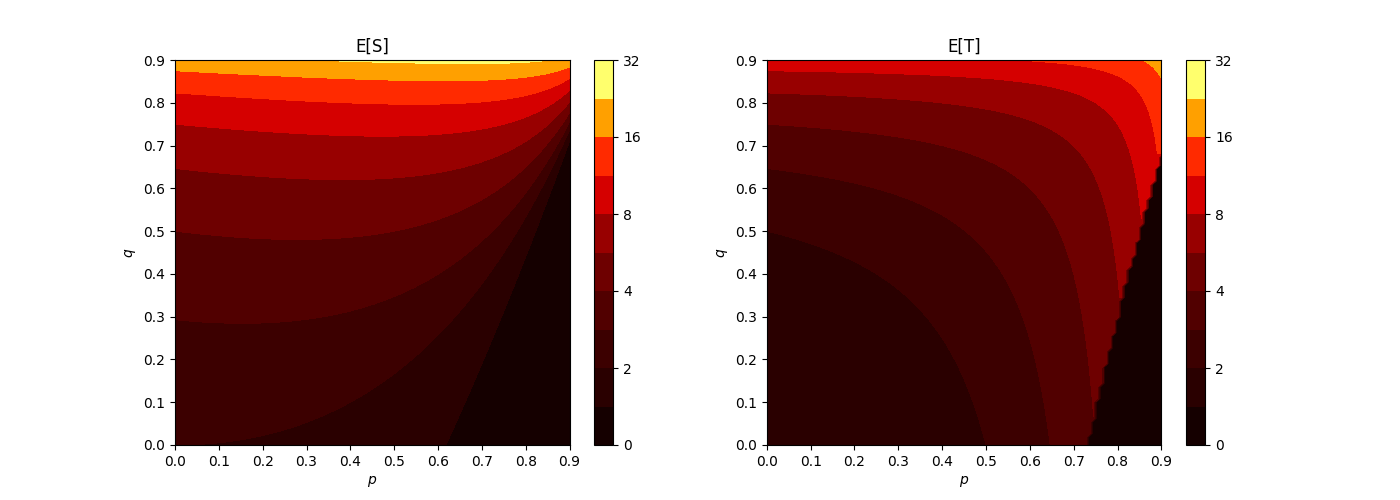}
  \caption{$\E{S}$ and $\E{T}$ plotted for the content manifold $\mc M_3$ from
  Example~\ref{ex:both} with $v_0 = 3$ and $\alpha = 1$.}%
  \label{fig:c_s_heatmpap_vp_vm_log}
\end{figure}

Our aim here is to reason about the impacts of a platform choosing content
within a particular content manifold $\mc M$. Let $\omega_S, \omega_T \in \mc M$ be the
parameters that maximize $\E{S}$ and $\E{T}$ respectively. Ideally, a platform
would like to choose content with parameters $\omega_S$; however, because by
assumption they
only observe signals about engagement, they optimize for $T$ and choose
$\omega_T$ instead. In Example~\ref{ex:quality}, this was a good strategy:
maximizing $\E{T}$ led to maximal $\E{S}$. But as we saw in
Examples~\ref{ex:moreishness} and~\ref{ex:both}, there is no guarantee that
leads to content with high utility for the user. Why is this the case, and when
is engagement a good proxy for utility?

\subsection{Characterizing the misalignment between utility and engagement}

We will show that there are two distinct reasons why utility and
engagement may be maximized at different points in a content manifold:
\begin{enumerate}
  \item Engagement-maximizing content has higher moreishness than
    utility-maximizing content.
  \item Engagement-maximizing content has higher span but lower value than
    utility-maximizing content.
\end{enumerate}
These are both intuitive reasons why maximizing engagement may not lead to
maximal user utility. Content may be more engaging because it keeps the user
``hooked,'' or it may be more engaging because it takes longer for users to
accomplish what they came to do (like a long series of instructional videos
where a simple, short one would have sufficed).

In order to make this intuition precise, we need the following definitions.
For $\omega = (p, q, \vbar)$, define
\begin{align*}
  g_S(\omega) &\triangleq \frac{\vbar - W}{1 - q} - \frac{pW}{1 - p} \\
  g_T(\omega) &\triangleq \frac{1}{1 - q} + \frac{p}{1 - p}
\end{align*}
Observe that
\begin{align}
  \omega_S &= \argmax_{\omega \in \mc M} g_S(\omega) \label{eq:wS_def} \\
  \omega_T &= \argmax_{\omega \in \mc M} g_T(\omega) ~ ~ ~ \text{s.t.} ~
  g_S(\omega) \ge 0 \label{eq:wT_def}
\end{align}
Assuming ties are broken consistently, these are well-defined because $\mc M$ is
closed. We'll assume that $\mc M$ contains at least one point with strictly
positive $\E{S}$. Let $\omega_S = (p_S, q_S, \vbar_S)$ and $\omega_T = (p_T,
q_T, \vbar_T)$. Then, the following theorem characterizes the two reasons why
utility and engagement can be maximized by different points on the content
manifold.

\def\pT{p_T}
\def\qT{q_T}
\def\vT{\vbar_T}
\def\pS{p_S}
\def\qS{q_S}
\def\vS{\vbar_S}

\begin{theorem}
  \label{thm:two_conds}
  For a given content manifold $\mc M$, suppose $\omega_S \ne \omega_T$, and
  that the disagreement is strict: $g_S(\omega_S) > g_S(\omega_T)$ and
  $g_T(\omega_T) > g_T(\omega_S)$.
  One of the following conditions must hold:
  \begin{enumerate}
    \item $\pT > \pS$ ($\omega_T$ has higher moreishness)
      \label{item:high_p}
    \item $\qT > \qS$ and $\vT < \vS$
      ($\omega_T$ has higher span but lower value) \label{item:high_q_low_v}
  \end{enumerate}
\end{theorem}
\begin{proof}
  To prove the theorem, it suffices to show that when
  Condition~\ref{item:high_p} doesn't hold, Condition~\ref{item:high_q_low_v}
  does. Assume that $\pT \le \pS$. First we will show that this
  implies $\qT > \qS$.
  To do so, observe that
  \begin{align*}
    g_T(\omega_T)
    &> g_T(\omega_S) \\
    \frac{1}{1-\qT} + \frac{1}{1-\pT} - 1
    &> \frac{1}{1-\qS} + \frac{1}{1-\pS} - 1 \\
    \frac{1}{1-\qT}
    &> \frac{1}{1-\qS} \tag{$\pT \le \pS$} \\
    \qT
    &> \qS
  \end{align*}
  Next, to show that $\vT < \vS$, we have
  \begin{align*}
    g_S(\omega_S)
    &> g_S(\omega_T) \\
    \frac{\vS - W}{1-\qS} - \frac{W}{1-\pS} + W
    &> \frac{\vT - W}{1-\qT} - \frac{W}{1-\pT} + W \\
    \frac{\vS - W}{1-\qS}
    &> \frac{\vT - W}{1-\qT} \tag{$\pT \le \pS$} \\
    \frac{\vS - W}{1-\qS}
    &> \frac{\vT - W}{1-\qS} \tag{$\qT > \qS$} \\
    \vS
    &> \vT
  \end{align*}
  This proves the desired claim: whenever $\pT \le \pS$,
  $\qT > \qS$ and $\vT < \vS$.
\end{proof}

\def\wsrp{(\omega_S + c\vr)[p]}
\def\wsrq{(\omega_S + c\vr)[q]}
\def\wsrv{(\omega_S + c\vr)[\barv]}

\subsection{Content manifolds for which utility and engagement are aligned}

Given this general misalignment between utility and engagement, we might
rightly ask: why would the platform try to optimize for engagement in the first
place? In fact, there are natural assumptions under which this is a fairly
reasonable thing to do. But crucially, these assumptions are on the structure of
the content manifold---in other words, whether or not engagement-maximization is a good
strategy to improve user welfare depends on our beliefs about the shape of the
content manifold. One example of such assumptions is the following:
\begin{enumerate}
  \item No content is very moreish.
  \item All content has roughly equal value.
\end{enumerate}
If these are both true, then maximizing engagement leads to near-optimal
user utility. In other words, under these conditions, a content manifold is
``salad''-like. We can formalize this claim as follows.

\begin{theorem}
  \label{thm:utility_close}
  Suppose that a content manifold $\mc M$ satisfies the following conditions:
  \begin{enumerate}
    \item $p \le \alpha$ for all $(p, q, \barv) \in \mc M$
    \item For some $v_0$, $|\barv - v_0| < \beta$ for all $(p, q, \barv) \in \mc
      M$
  \end{enumerate}
  Then, the user's utility at $\omega_T$ is not much less than its utility at
  $\omega_S$:
  \begin{align*}
    \E{S(\omega_T)} \ge \E{S(\omega_S)}  - 2 \beta \E{T(\omega_S)} -
    \frac{\alpha(v_0+\beta)}{1-\alpha}.
  \end{align*}
  Thus, when $\alpha$ and $\beta$ are close to 0, the user's utility under
  engagement-maximization is near-optimal.
\end{theorem}
\begin{proof}
  Define
  \begin{align*}
    a_i &\triangleq \frac{1}{1-p_i} \tag{$i \in \{S, T\}$} \\
    b_i &\triangleq \frac{1}{1-q_i} \tag{$i \in \{S, T\}$}
  \end{align*}
  Because $\E{T(\omega_T)} > \E{T(\omega_S)}$,
  \begin{align*}
    \frac{1}{1-\pT} + \frac{1}{1-\qT} - 1
    &> \frac{1}{1-\pS} + \frac{1}{1-\qS} - 1 \\
    a_T + b_T
    &> a_S + b_S \\
    b_S - b_T
    &< a_T - a_S
    \numberthis \label{eq:T_comp}
  \end{align*}
  Note that by assumption,
  \begin{equation}
    \label{eq:vt_bound}
    \vT \le v_0 + \beta,
  \end{equation}
  and for $i \in \{S, T\}$,
  \begin{equation}
    \label{eq:a_bound}
    1 \le a_i \le \frac{1}{1-\alpha}.
  \end{equation}
  Using this,
  \begin{align*}
    \E{S(\omega_S)} - \E{S(\omega_T)}
    &= \frac{\vS - W}{1 - \qS} - \frac{W}{1-\pS}
    - \b{\frac{\vT - W}{1 - \qT} - \frac{W}{1-\pT}} \\
    &= b_S(\vS - W) - b_T(\vT - W) + W(a_T - a_S) \\
    &= b_S (\vS - \vT) + (b_S - b_T)(\vT - W) + W(a_T -
    a_S) \\
    &< b_S (\vS - \vT) + (a_T - a_S)(\vT - W) + W(a_T -
    a_S) \tag{by~\eqref{eq:T_comp}} \\
    &= b_S (\vS - \vT) + (a_T - a_S)\vT \\
    &\le b_S \cdot 2\beta + \frac{\alpha (v_0+\beta)}{1-\alpha}
    \tag{by~\eqref{eq:vt_bound} and~\eqref{eq:a_bound}} \\
    &\le 2\beta(a_S + b_S - 1) + \frac{\alpha (v_0+\beta)}{1-\alpha}  \\
    &= 2\beta \E{T(\omega_S)} + \frac{\alpha (v_0+\beta)}{1-\alpha} 
  \end{align*}
\end{proof}
Note that this gap must necessarily increase linearly with engagement $T$,
since the user can miss out on a constant utility per step if $\vT < \vS$.

\section{How Platforms might Learn the Type of Content Manifold They're On}
\label{sec:strategies}

At this point, a platform designer might rightly ask: if engagement-maximization
is a good strategy for some content manifolds but not others, how might they go
about determining what type of content manifold they have? In other words, how could
they distinguish a ``junk food'' manifold from a ``salad'' manifold, or restrict
themselves to a ``salad''-like portion of the manifold?
Theorem~\ref{thm:utility_close} characterizes content manifolds over in
$\Omega$, but in
reality, platforms optimize over some observable feature space $\mc X$ with no
access to $\Omega$. To some extent this is a platform-dependent activity, but
our model naturally suggests a few general strategies that can help tease apart
these different content manifold types. We describe three high-level approaches here:
user satisfaction surveys, value-driven data, and UI design choices. This list
is not meant to be exhaustive, but to provide examples of how a platform might
seek to better understand the content manifold(s) on which it operates through
extensions of this model.

\subsection{Surveys}

Recognizing that user utility can't be fully inferred from observational data,
many platforms run surveys to assess user satisfaction
(e.g.,~\citep{gesenhues2018facebook}). These surveys often ask users for
immediate feedback on a session or piece of content: how happy is the user
with how they just spent their time? In general, these surveys produce orders
of magnitude less data than the engagement data platforms collect through
observation, making it difficult to use survey outcomes as the sole measure of
how platform changes impact user satisfaction. Instead of using these data for
optimization, however, the platform could use survey data to determine whether
their content manifold is more ``junk food''-like or ``salad''-like. As a simple
example, they could consider the correlation between engagement and survey
outcomes: if changes that increase engagement tend to decrease utility, the
platform might need to reduce their reliance on engagement metrics. On the other
hand, if engagement and survey outcomes are positively correlated, then the
platform might assume that the content manifold they're assessing has relatively
well-aligned utility and engagement, allowing them to leverage the full power of
their behavioral data to optimize the platform.

Beyond a broad look at the correlation between engagement and measured utility,
our model suggests more sophisticated ways to use survey data to better
understand the characteristics of content. Session lengths are heterogeneous
across users and time; do longer sessions lead to higher user utility? And what
should this tell us about overall welfare? Concretely, we can frame this
within our model as follows by assuming that a post-session survey provides a
(perhaps noisy) estimate of utility for that session. Then, we can ask the
following: does $\E{S \given T = t}$ increase with $t$? The following theorem
shows how a platform might use the empirical relationship between utility and
session length to better understand its content.

\begin{theorem}
  \label{thm:length_surplus}
  There are three regimes of interest:
  \begin{enumerate}
    \item $q = p$. Then, $\E{S \given T = t}$ increases with $t$ if and only if
      $\barv > 2W$.
    \item $q > p$. Then,
      \begin{enumerate}
        \item If $\barv > 2W$, $\E{S \given T = t}$ increases with $t$ for all $t$.
        \item There exists $t^*$ such that $\E{S \given T = t}$ increases with
          $t$ for $t > t^*$.
      \end{enumerate}
    \item $q < p$. Then,
      \begin{enumerate}
        \item If $\barv < 2W$, $\E{S \given T = t}$ decreases with $t$ for all $t$.
        \item There exists $t^*$ such that $\E{S \given T = t}$ decreases with
          $t$ for $t > t^*$.
      \end{enumerate}
  \end{enumerate}
\end{theorem}
\begin{proof}
  Let $T_q$ be the first time where $I_t = 0$. First, observe that
  \begin{equation}
    \label{eq:S_tq}
    \E{S \given T = t} = \barv \cdot \E{T_q \given T = t} - Wt.
  \end{equation}
  To find $\E{T_q \given T = t}$, recall that $T_q \sim \text{Geom}(1-q)$ and $T -
  T_q + 1 \sim \text{Geom}(1-p)$. Let $T_p \triangleq T - T_q + 1$. Then,
  \begin{align*}
    \E{T_q \given T = t}
  &= \E{T_q \given T_q + T_p = t + 1} \\
  &= \sum_{\tau = 1}^t \tau \Pr[T_q = \tau \given T_q + T_p = t+1] \\
  &= \sum_{\tau = 1}^t \tau \frac{\Pr[T_q = \tau \cap T_p = t - \tau +
  1]}{\Pr[T_q + T_p = t+1]} \\
  &= \sum_{\tau = 1}^t \tau \frac{\Pr[T_q = \tau \cap T_p = t - \tau +
  1]}{\sum_{\tau' = 1}^{t} \Pr[T_q = \tau' \cap T_p = t - \tau' + 1]} \\
  &=  \frac{\sum_{\tau = 1}^t \tau\Pr[T_q = \tau \cap T_p = t - \tau +
  1]}{\sum_{\tau' = 1}^{t} \Pr[T_q = \tau' \cap T_p = t - \tau' + 1]}
  \numberthis \label{eq:Tq_Tp_frac}
  \end{align*}
  If $p = q$, then $\E{T_q \given T_q + T_p = t+1} = \frac{t+1}{2}$ by symmetry.
  Thus, we need only consider the case where $p \ne q$. For any $\tau \in \{1,
  \dots, t\}$,
  \begin{align*}
    \Pr[T_q = \tau \cap T_p = t - \tau + 1]
  &= q^{\tau-1}(1-q) p^{t - \tau} (1-p) \\
  &= (1-q)(1-p) p^{t+1} \p{\frac{q}{p}}^{\tau-1}
  \end{align*}
  Thus,
  \begin{align*}
    \sum_{\tau = 1}^t \Pr[T_q = \tau \cap T_p = t - \tau + 1]
  &= (1-q)(1-p) p^{t+1} \sum_{\tau=1}^t \p{\frac{q}{p}}^{\tau-1} \\
  &= (1-q)(1-p) p^{t+1} \sum_{\tau=0}^{t-1} \p{\frac{q}{p}}^{\tau} \\
  &= (1-q)(1-p) p^{t+1} \p{\frac{1 - (q/p)^t}{1 - q/p}} \\
  &= (1-q)(1-p) p^{t+1} \p{\frac{\beta^t - 1}{\beta - 1}} \tag{$\beta \triangleq
  q/p$}
  \end{align*}
  Similarly, defining $\beta \triangleq \frac{q}{p}$,
  \begin{align*}
    \sum_{\tau = 1}^t \tau \Pr[T_q = \tau \cap T_p = t - \tau + 1]
  &= (1-q)(1-p) p^{t+1} \cdot \sum_{\tau=1}^t \tau
  \p{\frac{q}{p}}^{\tau-1} \\
  &= (1-q)(1-p) p^{t+1} \cdot \frac{p}{q} \sum_{\tau=1}^t \tau
  \p{\frac{q}{p}}^{\tau} \\
  &= (1-q)(1-p) p^{t+1} \cdot \beta^{-1} \sum_{\tau=1}^t \tau
  \beta^\tau \\
  &= (1-q)(1-p) p^{t+1} \cdot \beta^{-1} \cdot \frac{\beta (t \beta^{t+1} -
  (t+1) \beta^t  + 1)}{(\beta-1)^2} \\
  &= (1-q)(1-p) p^{t+1} \cdot \frac{t \beta^{t+1} -
  (t+1) \beta^t  + 1}{(\beta-1)^2}
  \end{align*}
  Returning to~\eqref{eq:Tq_Tp_frac},
  \begin{align*}
    \E{T_q \given T_q + T_p = t + 1}
    &= \frac{\sum_{\tau = 1}^t \tau\Pr[T_q = \tau \cap T_p = t - \tau +
    1]}{\sum_{\tau' = 1}^{t} \Pr[T_q = \tau' \cap T_p = t - \tau' + 1]} \\
    &= \frac{(1-q)(1-p) p^{t+1} \cdot \frac{t \beta^{t+1} - (t+1) \beta^t  +
    1}{(\beta-1)^2}}{(1-q)(1-p) p^{t+1} \p{\frac{\beta^t - 1}{\beta - 1}}} \\
    &= \frac{t \beta^{t+1} - (t+1) \beta^t  + 1}{(\beta^t - 1)(\beta - 1)} \\
     &= \frac{t\beta^t(\beta - 1) - (\beta^t - 1)}{\p{\beta^t - 1}(\beta - 1)} \\
     &= \frac{t \beta^t}{\beta^t - 1} - \frac{1}{\beta - 1}
  \end{align*}
  With this, we can return to~\eqref{eq:S_tq}.
  \begin{align*}
    \E{S \given T = t}
    &= \barv \cdot \E{T_q \given T = t} - Wt \\
    &= \begin{cases}
      \barv \cdot \frac{t+1}{2} - W t & p = q \\
      \barv \p{\frac{t \beta^t}{\beta^t - 1} - \frac{1}{\beta - 1}} - W t & p \ne q
    \end{cases}
    \numberthis \label{eq:St_cases}
  \end{align*}

  \textbf{Case 1:} $p = q$. \\
  Then,~\eqref{eq:St_cases} increasing in $t$ if and only if
  \begin{align*}
    \frac{d}{dt} \barv \cdot \frac{t+1}{2} - Wt
    &> 0 \\
    \frac{\barv}{2} &> W
  \end{align*}
  This proves the theorem in this case.

  \textbf{Case 2:} $p > q$. \\
  For $p \ne q$,~\eqref{eq:St_cases} is increasing in $t$ if
  \begin{align*}
    \frac{d}{dt} \barv \p{\frac{t \beta^t}{\beta^t - 1} - \frac{1}{\beta - 1}} - W t
    &> 0 \\
    \frac{\beta^t (\beta^t - t \ln \beta - 1)}{(\beta^t - 1)^2}
    &> \frac{W}{\barv} \\
    \frac{\beta^t (\beta^t - 1)}{(\beta^t - 1)^2} - \frac{t \beta^t \ln
    \beta}{(\beta^t - 1)^2}
    &> \frac{W}{\barv} \\
    \frac{\beta^t}{\beta^t - 1} - \frac{t \beta^t \ln \beta}{(\beta^t - 1)^2}
    &> \frac{W}{\barv} \\
    \frac{\beta^t}{\beta^t - 1} \p{1 -  \frac{t \ln \beta}{\beta^t - 1}}
    &> \frac{W}{\barv} \\
  \end{align*}
  When $q > p$, $\beta > 1$. Let $x \triangleq \beta^t-1$. Then, $\E{S \given T
  = t}$ is increasing in $t$ if
  \begin{equation}
    \label{eq:q_ne_p_cond}
    \frac{x+1}{x} \p{1 - \frac{\ln (x+1)}{x}}
    > \frac{W}{\barv}.
  \end{equation}
  By Lemma~\ref{lem:half_lb}, for $x > 0$,
  \begin{align*}
    \frac{x+1}{x} \p{1 - \frac{\ln (x+1)}{x}}
    &\ge \frac{1}{2}.
  \end{align*}
  As long as $\barv > 2W$, $\E{S \given T = t}$ is increasing in $t$ as desired.
  Next, we must show that there exists $t^*$ such that $\E{S \given T = t}$ is
  increasing in $t$ for $t > t^*$. To do so, observe that $x = \beta^t - 1 \to
  \infty$ as $t \to \infty$. Thus, it suffices to show that in the
  limit,~\eqref{eq:q_ne_p_cond} holds.
  \begin{align*}
    \lim_{x \to \infty} \frac{x+1}{x} \p{1 - \frac{\ln (x+1)}{x}}
    &= \lim_{x \to \infty} \frac{x^2 + x - (x+1) \ln (x+1)}{x^2} \\
    &= \lim_{x \to \infty} \frac{2x + 1 - 1 - \ln(x+1)}{2x} \\
    &= \lim_{x \to \infty} \frac{2x - \ln(x+1)}{2x} \\
    &= \lim_{x \to \infty} \frac{2 - \frac{1}{x+1}}{2} \\
    &= 1 \\
    &> \frac{W}{\barv}
  \end{align*}
  since $\barv > W$ by assumption. Thus, there exists sufficiently large $x$ such
  that~\eqref{eq:q_ne_p_cond} holds, meaning the desired $t^*$ exists.

  \textbf{Case 3:} $q < p$. \\
  Finally, we consider the case where $q < p$, meaning $\beta < 1$. We can again
  use Lemma~\ref{lem:half_lb} with $x = \beta^t - 1 \in (-1, 0)$, meaning
  \begin{align*}
    \frac{x+1}{x} \p{1 - \frac{\ln (x+1)}{x}}
    &\le \frac{1}{2}.
  \end{align*}
  Thus, if $\barv < 2W$, $\E{S \given T = t}$ is decreasing in $t$. To show the
  existence of $t^*$ such that $\E{S \given T = t}$ is decreasing in $t$ for $t
  > t^*$, we use the fact that $x \to -1$ as $t \to \infty$. In the limit,
  \begin{align*}
    \lim_{x \to -1^+} \frac{x+1}{x} \p{1 - \frac{\ln (x+1)}{x}}
    &= \lim_{x \to -1^+} -\frac{(x+1) \ln(x+1)}{x^2} \\
    &= \lim_{x \to -1^+} -\frac{\ln(x+1)}{\frac{x^2}{x+1}} \\
    &= \lim_{x \to -1^+} -\frac{\frac{1}{x+1}}{\frac{x(x+2)}{(x+1)^2}} \\
    &= \lim_{x \to -1^+} -\frac{1}{\frac{x(x+2)}{x+1}} \\
    &= 0 \\
    &< \frac{W}{\barv}
  \end{align*}
  Therefore, for sufficiently large $t$ and $q < p$, $\E{S \given T = t}$ is
  decreasing in $t$.
\end{proof}
Thus, if the platform observes that user utility initially increases with
session length but starts to decrease after a while, it could conclude that its
content has high value but also high moreishness, and might take steps to try to
help users limit their engagement.

Another way a platform might seek to learn about its position on a content manifold
might be to directly ask users about how much unwanted time they spend on the
platform, sometimes termed ``regretful use''~\citep{cho2021reflect}. For
instance, suppose a platform surveyed users to determine how much time they want
to spend on the platform, and compare this to the actual amount of time they
spend. In our model, users want to consume $1/(1-q)$ pieces of content in
expectation, but end up consuming $1/(1-q) + p/(1-p)$ instead. The platform
could then view regretful use as a measure of its moreishness $p$. Note that
minimizing regretful use is not necessarily a useful objective for the
platform---trivially, the user has no regretful use when it avoids the platform
altogether, and for content manifolds like the one in Example~\ref{ex:both},
user
utility might be maximized at a different point from where regretful use is
minimized. But quantifying regretful use can still help a platform determine
whether their content has high or low moreishness, allowing them to make
decisions about engagement-maximization accordingly.

\subsection{Value-driven data}

Platforms typically collect a vast amount of behavioral data, from time on
platform to actions (likes, comments, etc.) to frequency of use. While our model
only considers one such behavior (amount of content consumed), in principle,
these different behaviors convey different amounts of information about user
utility. We might naturally try to use these different signals to infer utility
under concrete assumptions about the underlying relationship between utility and
particular signals~\citep{milli2021optimizing}.

As a simple example, the platform might look not just at the lengths of users'
sessions but also at the number of users on the platform. Intuitively, if users
have long sessions, but fewer of them want to use the platform at all, this
might be an indication that the platform has high moreishness and provides low
utility to users. We can capture this nuance by considering user heterogeneity,
and in particular, users with heterogeneous outside options.

Concretely, if each user in a heterogeneous population has outside option $W
\sim \mc W$, let $U(\omega)$ be the event that a random user from the
platform's user base chooses to use the platform on a particular day
for content parameters $\omega$. A platform that only looks at engagement data
for users who choose to use the platform is effectively maximizing $\EE{W \sim
\mc W}{T(\omega) \given U(\omega) = 1}$; but including data about the number of
users on the platform ($\Pr_{W \sim \mc W}[U(\omega) = 1]$) or the total volume
of engagement ($\EE{W \sim \mc W}{T(\omega)}$ over the whole user base) would
give the platform information about whether increased engagement leads to higher
utility. And while these measures may be hard to interpret in terms of
the total number of users on the platform, if we imagine that the platform
currently has a population of users and measures relative changes in their
frequency of use, the platform can use trends over time or different
conditions of an A/B test to identify behavioral patterns.
Figure~\ref{fig:heterogeneous_consumption} demonstrates how these different
measures provide information that can help determine the structure of a content
manifold: as the platform varies $p$, high $\E{T \given U = 1}$ coupled
with low $\Pr[U = 1]$ indicates that the content manifold is ``junk food''-like.
In other words, the platform can distinguish between a change that drives
some users off the platform while increasing engagement for those who stay from
a change that increases both engagement and the number of active users.

\begin{figure}[ht]
  \centering
  \includegraphics[width=0.9\linewidth]{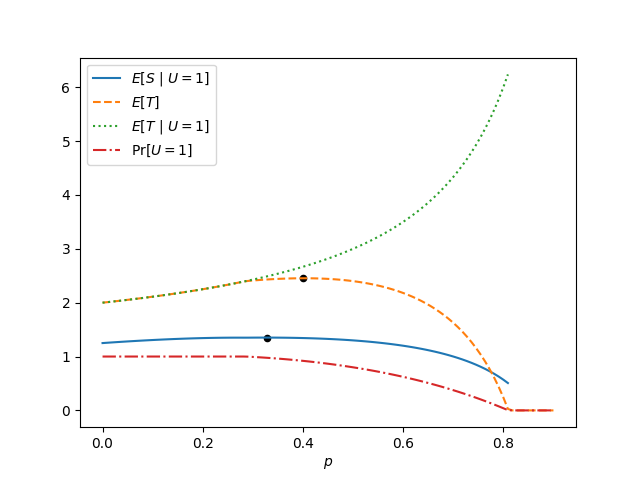}
  \caption{When $W \sim U[0.5, 1]$, different measures of engagement provide
    different insight into user utility for the content manifold $\mc M = \{(p, 0.5,
      \barv)
  : p \in [0, 1), \barv = 1.7p + (1-p)\}$.}%
  \label{fig:heterogeneous_consumption}
\end{figure}

\subsection{User interface and design choices}
\label{sec:UI}

While much of our discussion has been centered around \textit{content}
decisions, platforms also make \textit{design} decisions that significantly
impact user behavior. These user interface and design choices can also be
represented within our model in terms of $p$, $q$, and $\vbar$. In
Section~\ref{sec:model}, we considered how adding breaks to a user's content
feed might alter the content characteristics, reducing overall moreishness at
the expense of decreased value. In particular, we modeled a suggested
break as a piece of content with low moreishness, low value, and high span.
Thus, a user who encounters a break is likely to have system 2 activated,
allowing them to choose to leave the platform. Note that by
Proposition~\ref{prop:mixture} this equivalent to slightly updating the content
distribution to having lower $p$, lower $\vbar$, and higher $q$. We could also
imagine a feature like autoplay having the opposite effect on moreishness, since
it decreases the likelihood that system 2 will regain control.

Through a concrete model of the effects these features will have on content
parameters, we can use them to solve an \textit{inference} problem: when a
feature like suggested breaks or autoplay is introduced, how should we interpret
the resulting behavioral changes that we observe?
In the case of suggested breaks~\citep{instagram-breaks}, if the user is still
interested in consuming more content ($I_t = 1$), the break will have little
impact on the user's behavior; but if the user is only still on the platform
because system 1 is ``hooked'' ($I_t = 0$), the break will increase the
probability that system 2 regains control and chooses to leave the platform. On
a content manifold with low moreishness, this would have little impact on
engagement, but on a content manifold, with high moreishness, this would
dramatically reduce engagement (and increase utility in the process). And
moreover, if the platform observes that users often leave the platform
immediately following a break, this would be further evidence that the
platform's content is particularly moreish. Similarly, platforms can learn from
the behavior of users who voluntarily adopt interventions designed to promote
self-control: for example, if the number of users who voluntarily add suggested
breaks to their feeds increases as the platform optimizes for engagement, the
platform can infer that its optimization is increasing engagement at the expense
of utility. Thus, behavioral changes induced by design changes can be leveraged
to produce a more nuanced understanding of the underlying content manifold.

A platform could also use these design choices to make inferences about more
granular \textit{types} of content. Suppose that a platform introduces an
autoplay feature and observes that overall engagement increases, but in
particular, engagement with celebrity gossip videos increases much more than
engagement with educational science videos. The platform might reasonably
conclude that celebrity gossip videos are relatively more moreish than science
videos, and this might lead them to treat these types of content differently
when optimizing.

Similarly, consider the design of embedded media, where users directly share
content with one another through other platforms. A user may thus come across
content from the platform without actively seeking it out. In our model, we can
think of this as forcing the user to use the platform without giving system 2
the opportunity to initially refuse. If a platform observes that much of its
traffic for a particular type of content comes from embedded media, they might
again conclude that this content type has high moreishness. Importantly, these
conclusions rest on some intuition that the designer has for the effect of a
design decision on the underlying parameters of user behavior; they cannot be
inferred from data alone.

\section{Extending the Model to Incorporate User Content Choices}
\label{sec:tree}
Our results so far demonstrate the consequences of optimizing for engagement on
a linear content feed. While this describes many popular platforms, others allow
users to choose the next piece of content they consume from a limited choice
set. For example, with the autoplay feature turned off, YouTube presents a user
with a set of suggested videos to watch next. Graphically, we can model this
using a tree structure, where each node represents a piece of content, and each
branch out of a node points to one of $d$ pieces of content curated by the
platform.

In this framework, it is natural to consider the impact that the branching
factor $d$ has on the user's behavior. We might think of two opposing forces
that influence the optimal choice of $d$: (i) more choice (higher $d$) will lead
to higher utility and engagement, since users can choose the content they
prefer; but (ii) too much choice will overwhelm the user and reduce both utility
and engagement. This suggests that engagement-maximization is well suited to
select the optimal $d$, since it appears to be aligned with utility under these
competing forces. But our model suggests that there is a third factor that this
standard intuition fails to capture: (iii) more choice increases the likelihood
that system 1 is active, increasing engagement but potentially decreasing
utility. This would imply that maximizing engagement may lead to a higher $d$
than is optimal for user utility.

\xhdr{Extending the model}

In order to formalize this claim, we need to extend our model as defined for
linear feeds to tree-structured feeds. There are two important distinctions
between the linear and tree settings which will require slight modifications of
the model:
\begin{enumerate}
  \item Whenever the user finishes consuming a piece of content, they are
    forced to actively choose the next content. As a result, system 1 is
    interrupted: even if the previous video was highly appealing to system 1,
    the forced choice breaks system 1's activity. In our original model,
    system 1 was engaged or not based on the previous content; in this model,
    system 1's engagement depends on the characteristics of the \textit{next}
    piece of content to be consumed. Thus, the special case of $d=1$ does not reduce
    to our original model.
  \item On the other hand, the platform must show the user previews (such as a
    video thumbnail or article title) of the $d$ pieces of content, and those
    previews can activate system 1.
\end{enumerate}
Taken together, these differences imply that we should modify the model such
that whether or not system 1 is active at step $t$ depends not on the content
consumed at the previous step, but on the properties of the $d$ pieces of
content recommended for consumption at step $t+1$.\footnote{This also suggests
that a platform interested in maximizing engagement might try to design an
interface that avoids forcing a choice that might break system 1's activation
from the previous content while also using previews to further increase the
likelihood that system 1 is active. In fact, this is reasonable description of
YouTube's interface with autoplay turned on: the platform will automatically
start the next video after a set amount of time, and the user will not be forced to
make a choice. But YouTube also shows thumbnails of other videos, so even if
system 1 is not already active, these previews can also activate it.}

As before, we will assume that each piece of content has parameters $(p, q, v)$.
These parameters are not directly comparable to those in our linear-feed
model. We consider the parameters to be a function of both the content itself
and the platform design: enabling autoplay, for example, increases the value of
$p$ without changing the content itself.
Denote the parameters of the $d$ options presented to the user at step $t$ as
$\{(p_i, q_i, v_{i,t}): i \in \{1, \dots, d\}\}$. Formally, we will model the
user's choice at step $t$ as follows:
\begin{enumerate}
  \item First, the user scans the options in a random order, and if any appeals
    to system 1 (which happens independently with probability $p_i$ for each
    $i$), the user immediately selects that content.\footnote{We can substitute
      in a variety of different models of system 1 choice (e.g., scan items
      sequentially or choose the one that maximizes system 1 utility $u_{i,t}$);
      these lead to similar results, only changing $\barv$ as defined in
    Theorem~\ref{thm:tree_tau}.}
  \item If no option appeals to system 1 (which happens with probability
    $\prod_{i=1}^d (1-p_i)$), system 2 chooses either option that maximizes
    expected utility or leaves if all options yield negative expected utility.
    Importantly, the user observes $v_i$ for each $i$ when making this choice.
\end{enumerate}
For simplicity, we will assume that the content distribution at each node is
identical. In practice, this may not be true---clicking on a celebrity news
video is likely to lead to more celebrity news videos---and extending this model
to consider such correlations is an interesting direction for future work. We
will further assume that $q_i = q$ for all $i$, although this assumption can be
lifted with some additional work. Finally, we assume that each $v_{i,t} \sim \mc
V_i$ independently, where $\EE{v_{i,t} \sim \mc V_i}{v_{i,t}} = \barv_i$.

\subsection{Characterizing Behavior}
We begin by determining the user's behavior under this model. Note that because
the future content distribution is independent of the user's choice, system 2
will always choose the option that maximizes $v_i$ (or leave the platform).
First, we characterize the user's policy as a threshold on the maximal $v_i$.
(We prove this in Appendix~\ref{app:proofs}.)
\begin{theorem}
  \label{thm:tree_tau}
  Let $v_t = \max \{v_{i,t} : i \in \{1, \dots d\}\}$, and let $\barv =
  \p{\sum_{i=1}^d p_i \barv_i}/\p{\sum_{i=1}^d p_i}$. Under the tree model, the
  user's policy when system 2 is active is to choose a threshold $\tau^*$ and
  to choose the option with maximal value $v_t$ whenever $v_t \ge \tau^*$ and to
  leave the platform if no such option exists. If the user uses the platform,
  its utility is
  \begin{align*}
    \E{S}
    &= \hat p \p{\frac{\barv}{1 - \hat p q} - \frac{W}{1 - \hat p}}
    + \p{\frac{1 - \hat p}{1 - \hat p q}} \frac{\E{v_t \cdot
        \ind{v_t \ge \tau^*}} + \p{\frac{\vbar \hat p q}{1 - \hat p q} - \frac{W}{1
      - \hat p}} \Pr[v_t \ge \tau^*]}{1 - \p{\frac{q (1 - \hat p)}{1 - \hat p q}}
    \Pr[v_t \ge \tau^*]},
  \end{align*}
  where $\tau^*$ is chosen to maximize $\E{S}$ and can be approximated through
  binary search, and $\hat p = 1 - \prod_{i=1}^d (1-p_i)$. Still assuming that
  the user uses the platform, its engagement is
  \begin{align*}
    \E{T}
    &= \frac{\hat p}{1-\hat p} + \p{\frac{1}{1 - \hat p q}} \frac{\Pr[v_t \ge
    \tau^*]}{1 - \p{\frac{q (1 - \hat p)}{1 - \hat p q}} \Pr[v_t \ge \tau^*]}.
  \end{align*}
  If the expression for $\E{S}$ is negative, then the user doesn't use the
  platform, and $S = T = 0$.
\end{theorem}

\subsection{Platform Optimization}

\begin{figure}[ht]
  \centering
  \includegraphics[width=0.9\linewidth]{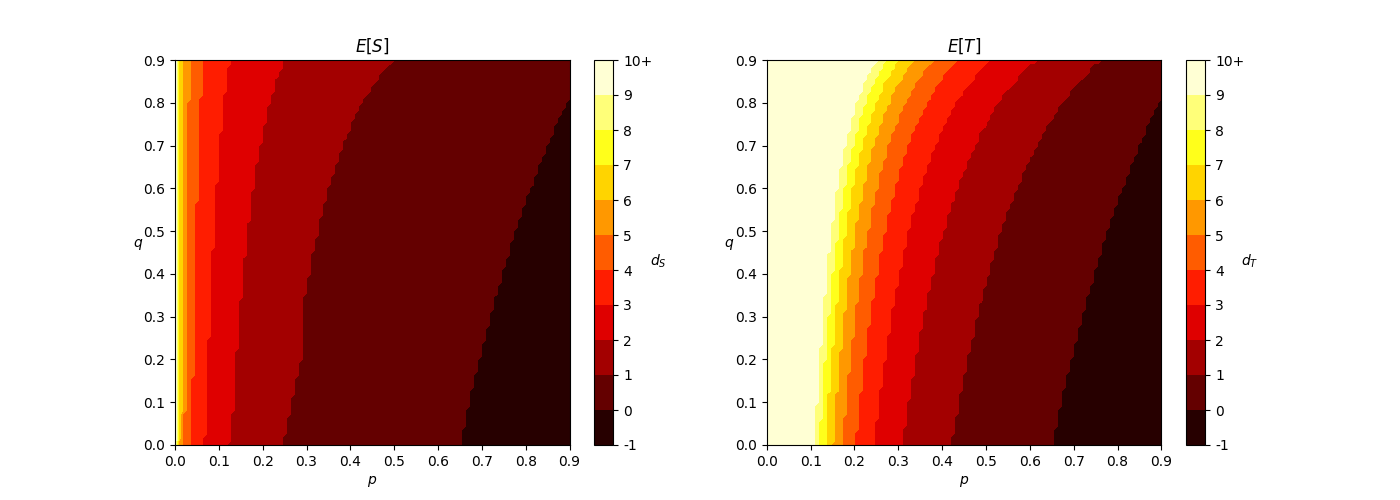}
  \caption{The choice of $d$ that maximizes $\E{S}$ and $\E{T}$ when each
    $v_{i,t}$ is $1$ with probability $0.5$ and $4$ with probability $0.5$. When
  the optimal $d$ is $0$, the user will not use the platform for any $d$.}%
  \label{fig:tree_heatmap}
\end{figure}

Based on this characterization, we can reason about the impact of system 1 on
platform design. Suppose the platform wants to determine $d^*$, the optimal
number of choices to present to a user. Doing this experimentally (e.g., through
A/B testing) would lead the platform to choose the $d_T$, defined to be the
choice of $d$ that maximizes $\E{T}$; however, there is no guarantee that this
will maximize the user's utility. Figure~\ref{fig:tree_heatmap} shows, for a
fixed value distribution $\mc V$, $d_S$ and $d_T$ that maximize $\E{S}$ and
$\E{T}$ respectively as a function of $p$ and $q$. Qualitatively, $p$ and $q$
have the same impact on $d_S$ and $d_T$: higher $p$ leads to lower $d$, and
higher $q$ leads to higher $d$. Intuitively, this is what we might expect. As
$d$ increases, the system 2 has better options from which to choose, but the
chance that system 2 is activated goes down since it is more likely that some
piece of content appeals to system 1. Thus, if $p$ is too high, the user will
refuse to use the platform altogether because $\hat p$, the overall probability
system 1 is active, increases exponentially in $d$. On the other hand, as $q$
increases, the user wants to spend more time on the platform, allowing them to
benefit more from the increase in value that higher $d$ brings.

While the two plots are directionally similar, they differ in that for fixed
parameters, $d_S$ tends to be lower than $d_T$. Again, this should match our
intuition: increasing $d$ should raise the likelihood that system 1 is active,
thereby increasing engagement. As a result, we would expect for
engagement-maximization to result in the maximal $d$ such that the user is
still willing to use the platform. This would imply that $d_S \le d_T$ in
general; however, there's a competing force that can lead to the opposite
outcome: increasing $d$ makes the user more discerning, using a higher
threshold $\tau^*$ to determine whether it wants to stay on the platform out of
fear of getting sucked in by system 1. Example~\ref{ex:tree} (in
Appendix~\ref{app:tree_ex}) demonstrates that 
it's possible to have $d_S > d_T$, meaning the user gets higher utility and
consumes less for a larger value of $d$. Engagement-maximization is
thus unpredictable, and its effects cannot always be simply characterized.

\section{Related Work}
to the questions of online platform design. Psychologically, our formalism draws
more on behavioral theories of time
inconsistency~\citep{akerlof1991procrastination}, multiple
selves~\citep{thaler1981economic,fudenberg2006dual}, and hyperbolic
discounting~\citep{laibson1997golden,o1999doing} rather than on dual processing
models~\citep{evans2008dual}. Our work draws upon models in both economics and
computer science, and in particular, models of naive and sophisticated
agents~\citep{o1999doing,kleinberg2014time,kleinberg2016planning}.
Related theoretical models study preference inconsistency in contexts
like productivity~\citep{makarov2011networking} and
privacy~\citep{liu2020data}. While our focus on online content consumption,
empirical evidence for inconsistent preferences has been found in a wide range
of settings including grocery
shopping~\citep{milkman2010ll,zatz2021comparing}, charitable
giving~\citep{cryder2017charity}, and movie
consumption~\citep{milkman2009highbrow}.

The relationship between behavior and users' true preferences online has
been studied both theoretically and empirically under a variety of names,
including behaviorism~\citep{ekstrand2016behaviorism}, digital
addiction~\citep{allcott2021digital}, regretful or problematic
use~\citep{cho2021reflect,shin2013automatically}, value
measurement~\citep{milli2021optimizing,lyngs2018so}, and
value-alignment~\citep{stray2021you}. Recent empirical studies document
self-control problems with social
media~\cite{allcott2021digital,hoong2021self}. In practice, some have sought to
address this issue by developing tools to help users practice self-control
(e.g.,~\citet{hiniker2016mytime}; see~\citet{cho2021reflect} for more examples).

Our model also relates to questions about how design choices can influence user
behavior and welfare, often studied in the HCI community. Examples include the
impacts of design on agency~\citep{lukoff2021design}, persuasive
technology~\citep{fogg2002persuasive}, dark patterns~\citep{gray2018dark}, and
behavioral change~\citep{sleeper2015would,lyngs2019self,lyngs2020just}. Our hope
is that our model can provide some insight into what types of design
interventions might help users, when they would be appropriate, and how to
measure their impacts.

\section{Conclusion}
In this work, we show the problems can arise when platforms
misunderstand the psychology of user behavior.  In the face of inconsistent
preferences, maximizing \textit{apparent} utility (i.e., engagement) fails to
maximize \textit{actual} utility.
Our characterization theorems show that the
misalignment between engagement and utility is not universal. For some platforms
and content distributions, engagement is a good proxy for utility; for others,
there are structured reasons why they diverge. Through the lens of preference
inconsistency, we analyzed how a platform might use additional sources of data
or design decisions to better understand how content properties interact with
engagement optimization.

Our work seeks to open up a dialogue between theoretical models, user interface
design, and social media platforms. We see several natural directions for future
work. First, a richer characterization of content manifolds will provide a
better understanding of how different types of content interact differently with
engagement optimization. Second, while we suggest a few ways one might elicit
and incorporate measures of user utility beyond engagement, how best to merge
engagement data with various other forms of information about user welfare
remains an open question. Third, our model leads to a number of empirical
questions regarding the properties of content on existing platforms and how they
differ, particularly across different users. And finally, while our model
treats users as static, recent work suggests that user adaptation to content
(e.g., habit formation) may play a major role in
overconsumption~\citep{allcott2021digital}. Future work could integrate models
of both platform and user adaptation to develop a more complete picture of
consumption.

There's a growing understanding that there's something not quite right with
platform optimization based on user behavioral data, even ignoring the privacy
concerns and financial incentives involved. Platforms appear to increasingly
realize that no matter how carefully they measure engagement, using it as a
metric doesn't quite lead to a product users are actually happy with. This paper
tries to address these issues by arguing that no matter how sophisticated the
model of engagement, it cannot be effective unless it acknowledges that users
have internal conflicts in their preferences. This realization has led to two
kinds of changes. On the design side, it has led to experiments with  UI changes
to see if they improve happiness (such as time limits). On the content
optimization side, it has resulted in attempts to augment passive user behavior
data with explicit survey measures of happiness or satisfaction. In both cases,
the approaches are largely based on intuition. Our model provides a framework
for thinking about why engagement optimization may be failing. It also provides
a systematic way to think about and possibly generate platform level UI changes,
akin to \citet{milli2021optimizing} and \citet{lyngs2019self}. Formalism has
another important benefit. Survey measures will always be far more expensive and
rarer than behavioral data. For platforms to succeed, they need some way to
harness the engagement data that are plentiful while not being misled by them.

Formal models that incorporate the richness of human psychology, such as ours,
can help platforms use survey data to better understand the consequences of
their choices and thereby how to improve their optimization. More generally,
algorithms built on human behavioral data in settings ranging from online
platforms to expert judgment inherit the behavioral biases and inconsistencies
of their human data providers. Applying formal models of human behavior to these
settings can help algorithm designers ``invert'' behavior to understand
underlying values and beliefs \citep{perspectives-psych}.

\paragraph*{Acknowledgments}
We thank Cristos Goodrow and the participants in the AI, Policy, and Practice
initiative at Cornell for valuable feedback. This work was supported by a
Vannevar Bush Faculty Fellowship, a Multidisciplinary University Research
Initiative grant [W911NF-19-0217], and funding from the MacArthur Foundation,
the Simons Foundation and the Center for Applied AI at Booth School of Business.

\bibliographystyle{plainnat}
\bibliography{refs}

\appendix

\section{A Stateful Utility Model}
\label{app:stateful}

Here, we derive results for the alternative model for $I_t$ presented in
Section~\ref{sec:gamma}.

\subsection{Characterizing behavior}

At any time $t$ where system 2 is active, the agent must choose between
remaining on the platform until system 2 is once again active or leaving the
platform. Let $T_1, T_2, \dots$ be the times at which system 2 is active. Then,
the agent chooses to remain on the platform at time $T_k$ if and only if
\begin{align*}
  \E{\sum_{t=T_k}^{T_{k+1}-1} v_t \gamma^t} \ge
  \E{\sum_{t=T_k}^{T_{k+1}-1} W}.
\end{align*}
Observe that $T_{k+1} - T_k$ is the waiting time for a Bernoulli event of
probability $1-p$. Let $B$ be distributed as $T_{k+1} - T_k$. Then, the agent's
condition for remaining on the platform is
\begin{align*}
  \E{\sum_{t=T_k}^{T_{k+1}-1} v_t \gamma^t}
  &\ge \E{\sum_{t=T_k}^{T_{k+1}-1} W} \\
  \E{\sum_{t=0}^{B-1} v_{t+T_k} \gamma^{t+T_k}}
  &\ge \E{\sum_{t=0}^{B-1} W} \\
  \vbar \gamma^{T_k} \E{\sum_{t=0}^{B-1} \gamma^t}
  &\ge W \E{B}.
\end{align*}
Note that for any $r \in (0, 1)$,
\begin{align*}
  \E{\sum_{t=0}^{B-1} r^t}
  &= \sum_{n=0}^{\infty} r^n \Pr[B-1 \ge n] \\
  &= \sum_{n=0}^{\infty} r^n p^n \\
  &= \frac{1}{1-pr}.
\end{align*}
Therefore, the agent will continue if and only if
\begin{align*}
  \vbar \gamma^{T_k}\E{\sum_{t=0}^{B-1} \gamma^t}
  &\ge \frac{W}{1-p} \\
  \frac{\vbar \gamma^{T_k}}{1-p\gamma} &> \frac{W}{1-p} \\
  \frac{\vbar \gamma^{T_k}}{W} &> \frac{1-p\gamma}{1-p}
  \numberthis \label{eq:gamma-part}
\end{align*}

\begin{proposition}
  \label{prop:t-decrease}
  As moreishness increases, the agent's decision to continue at time $t$ can
  only change from willing to unwilling.
\end{proposition}
\begin{proof}
  Observe that $\frac{1-p\gamma}{1-p}$ is increasing in $p$:
  \begin{align*}
    \frac{1-p\gamma}{1-p}
    &= \frac{1-\gamma + \gamma - p\gamma}{1-p} \\
    &= \frac{1-\gamma}{1-p} + \frac{\gamma(1-p)}{1-p} \\
    &= \frac{1-\gamma}{1-p} + \gamma
  \end{align*}
  Thus, as $p$ increases, the right hand side of~\eqref{eq:gamma-part}
  increases monotonically, which proves the claim.
\end{proof}

The agent's behavior is also monotone in $T_k$: if the agent is willing to
continue at some $T_k = t$, then they are also willing to continue for any $t' <
t$, since the left hand side of~\eqref{eq:gamma-part} is decreasing in $T_k$.

\xhdr{Expected consumption}

Let $t^*$ be the largest integer such that
\begin{align*}
  \frac{\vbar \gamma^{t^*}}{W}
  &\ge \frac{1-p\gamma}{1-p} \\
  \gamma^{t^*}
  &\ge \frac{W(1-p\gamma)}{\vbar(1-p)} \\
  t^*
  &\le \frac{\log\p{\frac{W(1-p\gamma)}{\vbar(1-p)}}}{\log
  \gamma}. \tag{$\log \gamma < 0$}
\end{align*}
Then,
\begin{equation*}
  t^* \triangleq \left \lfloor
    \frac{\log\p{\frac{W(1-p\gamma)}{\vbar(1-p)}}}{\log
  \gamma} \right \rfloor.
\end{equation*}

Note that $t^*$ can be negative meaning the agent wouldn't use the platform at
all. The agent's behavior has a similar structure to our previous model: the
agent will always continue (whether system 1 or 2 is in control) for any $t \le
t^*$. For any $t > t^*$, the agent will leave the platform if system 2 is
active. The time that the agent remains on the platform beyond $t^*$ is again
the waiting time for a Bernoulli even with probability $1-p$.
Overall consumption is thus
\begin{equation}
  \E{T}
  = \begin{cases}
    t^* + \frac{1}{1-p} & t^* \ge 0 \\
    0 & \text{otherwise}
  \end{cases}.
  \label{eq:ET-gamma}
\end{equation}

While Proposition~\ref{prop:t-decrease} implies that $t^*$ is weakly decreasing
in $p$, expected consumption $\E{T}$ is not necessarily monotonic. Instead, it
has a sequence of increasing segments punctuated by discrete drops when $t^*$
decreases as shown in Figure~\ref{fig:gamma-model}.

\begin{proposition}
  Excluding discontinuities in $t^*$, $\E{T}$ is strictly increasing in $p$.
\end{proposition}

\begin{proof}
  To see this, observe that if
  we ignore the floor in the definition of $t^*$, we have

  \begin{align*}
    \frac{d}{dp} \p{\frac{\log\p{\frac{W(1-p\gamma)}{\vbar(1-p)}}}{\log
    \gamma} + \frac{1}{1-p}}
  &= \frac{d}{dp} \p{\frac{\log W - \log \vbar + \log(1-p\gamma) -
  \log (1-p)}{\log \gamma} + \frac{1}{1-p}} \\
  &= \frac{d}{dp} \p{\frac{\log(1-p\gamma) - \log (1-p)}{\log
  \gamma} + \frac{1}{1-p}} \\
  &= \frac{1}{\log \gamma} \p{\frac{-\gamma}{1-p\gamma} +
  \frac{1}{1-p}} + \frac{1}{(1-p)^2} \\
  &= \frac{1}{\log \gamma}
  \p{\frac{1-\gamma}{(1-p\gamma)(1-p)}} + \frac{1}{(1-p)^2}.
  \end{align*}
  Note that $1-p\gamma > 1 - p$ and $\log \gamma < 0$. Therefore,
  \begin{align*}
    \frac{1}{\log \gamma} \p{\frac{1-\gamma}{(1-p\gamma)(1-p)}} +
    \frac{1}{(1-p)^2}
    &> \frac{1}{\log \gamma} 
    \p{\frac{1-\gamma}{(1-p)^2}} + \frac{1}{(1-p)^2} \\
    &= \frac{1}{(1-p)^2} \p{1 + \frac{1-\gamma}{\log \gamma}}.
  \end{align*}
  Thus, it suffices to show that $\frac{1-\gamma}{\log \gamma} \ge -1$. Using the
  inequality $\log (1+x) \le x$ for $x > -1$, we have
  \begin{align*}
    \log(1+(\gamma-1)) &\le \gamma-1 \\
    \log \gamma & \le \gamma-1 \\
    -\log \gamma & \ge 1-\gamma \\
    -1 & \le \frac{1-\gamma}{\log \gamma}. \tag{$\log \gamma < 0$}
  \end{align*}
\end{proof}

\begin{figure}[t]
  \includegraphics[width=0.45\linewidth]{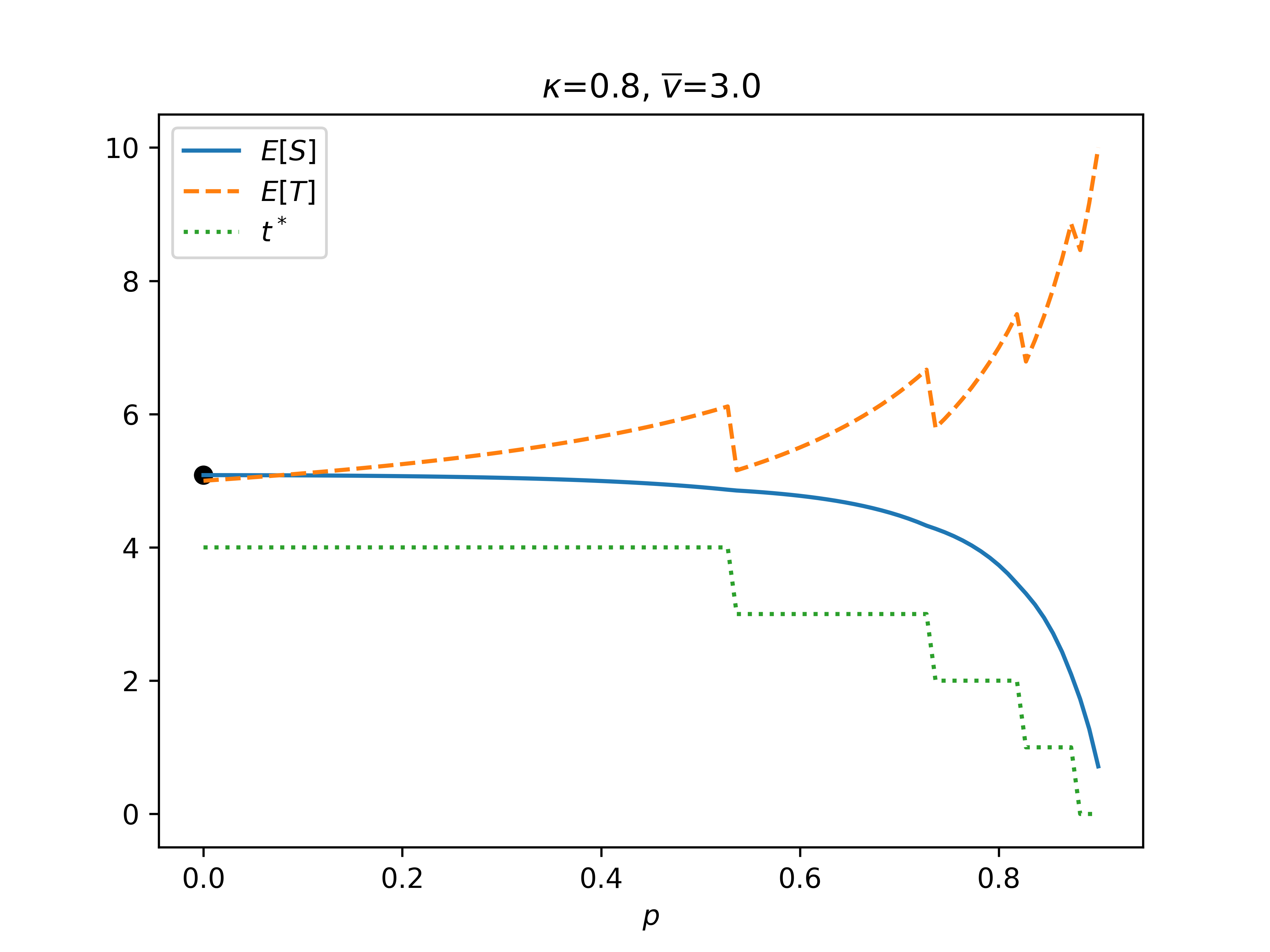}
  \includegraphics[width=0.45\linewidth]{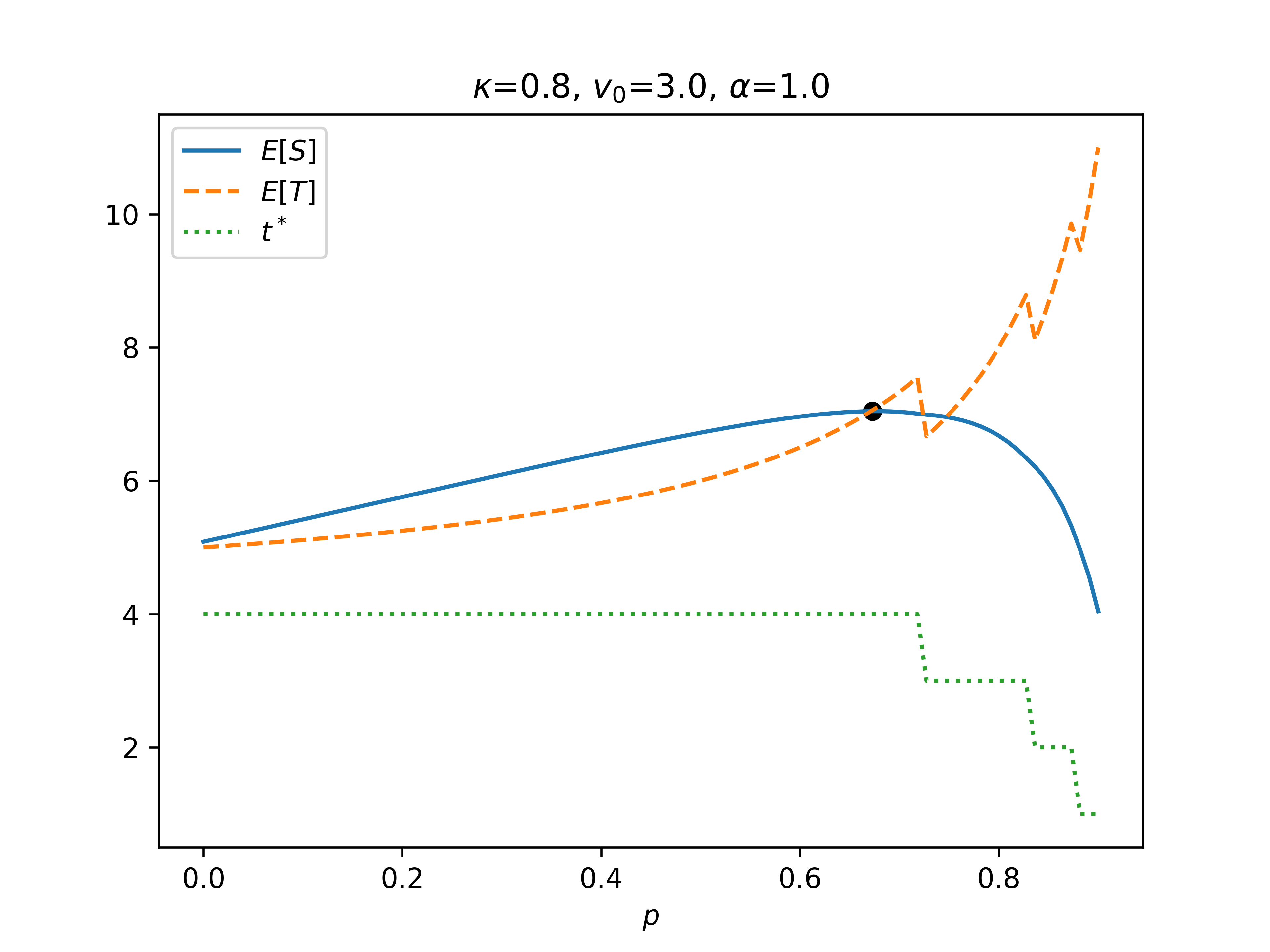}
  \caption{Consumption and utility for similar parameters to
  Examples~\ref{ex:moreishness} and~\ref{ex:both}}%
  \label{fig:gamma-model}
\end{figure}

\xhdr{Expected utility}

The agent's expected utility is the value they derive from consuming content
minus their outside option. Assuming $t^* \ge 0$, this is
\begin{align*}
  \E{S}
  &= \E{\sum_{t=0}^{T-1} \gamma^t v_t - W} \\
  &= \vbar \E{\sum_{t=0}^{T-1} \gamma^t} - W \E{T} \\
  &= \vbar \p{\sum_{t=0}^{t^*-1} \gamma^t + \E{\sum_{t=t^*}^{T-1} \gamma^t}} -
  W \p{t^* + \frac{1}{1-p}} \\
  &= \vbar \p{\frac{1-\gamma^{t^*}}{1-\gamma} + \gamma^{t^*} \E{\sum_{t=0}^{B-1}
  \gamma^t}} - W \p{t^* + \frac{1}{1-p}} \\
  &= \vbar \p{\frac{1-\gamma^{t^*}}{1-\gamma} + \gamma^{t^*} \frac{1}{1 - p
  \gamma}} - W \p{t^* + \frac{1}{1-p}}.
  \numberthis \label{eq:ES-gamma}
\end{align*}

In Figure~\ref{fig:gamma-model}, $\E{S}$ appears to be much smoother than
$\E{T}$. This is because each time $t^*$ decreases by 1, the agent is
approximately indifferent between continuing on the platform and leaving at time
$t^*$. Formally, let $T_f$ be that last time at which the agent chooses to
continue. Then,
\begin{align*}
  \E{S} &= \sum_{t=0}^{t^*} \E{S \given T_f = t}.
\end{align*}
As the agent becomes indifferent between remaining and leaving at $t^*$, $\E{S
\given T_f = t^*} = 0$. This means that
\begin{equation}
  \label{eq:indif-t}
  \sum_{t=0}^{t^*} \E{S \given T_f = t}
  = \sum_{t=0}^{t^*-1} \E{S \given T_f = t}.
\end{equation}
Let $\omega_0$ be a point on the content manifold given by $(p, \gamma, \vbar)$ at
which there is a discontinuity in $t^*$. In an arbitrarily small neighborhood
around $\omega$, it can be take on at most two distinct values: $t_1^*$ and
$t_2^* = t_1^* - 1$. In either case, by~\eqref{eq:indif-t},
\begin{align*}
  \lim_{\omega \to \omega_0} \E{S}
  &= \sum_{t=0}^{t_2^*} \E{S \given T_f = t} \bigg|_{\omega = \omega_0}
\end{align*}
Thus, $\E{S}$ is continuous over the content manifold, even though $\E{T}$ is
not. Note that derivatives of $\E{S}$ may not be continuous, since the addition
or removal of the $\E{S \given T_f = t^*}$ term discontinuously affects
derivatives.
\section{Equivalent Model Formulation}
\label{app:equiv}

Here, we present an equivalent formulation of diminishing returns.
We assume the agent has a capacity $C$, and that each piece
of content has a known size $s_t$.
The agent derives value $v_t$ from
consuming content if its capacity has not yet been filled, and 0 otherwise.
Formally, this is
\begin{align*}
  u_t = \begin{cases}
    v_t & \sum_{t' < t} s_t \le C \\
    0 & \text{otherwise}
  \end{cases}
\end{align*}
Let $I_t$ be the indicator for whether $\sum_{t' < t} s_t \le C$. We can
use this to write $u_t = v_t I_t$.

We assume that $C$ is drawn from an exponential distribution, i.e., $\Pr[C > c]
= e^{-\lambda c}$, but the agent never observes $C$ until its capacity has been
exceeded. Without loss of generality, assume $\lambda = 1$. Note that the under
this formulation, the agent is memoryless: for any $t$, conditioned on $I_t =
1$, $C$ is still exponentially distributed. As before, the agent has an outside
option with some fixed value $W$, so its net welfare for consuming content at
time $t$ is $u_t - W = v_t I_t - W$.

Given that $I_t = 1$, the probability that $I_{t+1} = 0$ is $e^{-s_t}$. Thus,
setting $q_t = e^{-s_t}$, this model resembles our original one. To show that
they are equivalent when $s_t$ comes from a known distribution but is
unobservable to the agent before consumption, let $q = \E{e^{-s_t}}$. Then, we
can see that regardless of the agent's consumption history, $\Pr[I_{t+1} = 1
\given I_t = 1] = q$, making this identical to the original model.

A related but slightly different formulation of diminishing returns is to assume
that the agent's utility at time $t$ for consuming content is $I_t v_t - W$ as
before, but with $I_t$ as a deterministic and decreasing sequence in $t$ (e.g.,
$I_t = \gamma^t$ for some $\gamma \in [0, 1)$). Analysis of this model yields
some qualitatively similar results, though its determinism makes it slightly
less clean to analyze.

\section{Deferred Proofs}
\label{app:proofs}

\begin{proof}[Proof of Theorem~\ref{thm:tree_tau}]
Suppose system 2 is active at step $t$. If $I_t = 0$, meaning
the agent has satiated, the agent will leave the platform since it cannot get
any utility from staying. Thus, to characterize behavior, we need only reason
about the case where $I_t = 1$. Let $S_t$ denote the agent's expected utility
for its remaining time on the platform beginning with step $t$. Let $Y_t = 1$ if
system 2 is active at time $t$ and $Y_t = 0$ otherwise. The agent's expected
utility is
\begin{align*}
  \E{S}
  &= \hat p \E{S_0 \given I_0 = 1, Y_0 = 0} + (1-\hat p) \E{S_0 \given I_0
  = 1, Y_0 = 1} \\
  &= \hat p \E{S_t \given I_t = 1, Y_t = 0} + (1-\hat p) \E{S_t \given I_t
  = 1, Y_t = 1}
  \numberthis \label{eq:eS_breakdown}
\end{align*}
since system 1 is active with probability $\hat p$ and the agent is memoryless.
To find the first term, let $N(t)$ be the next time system 2 is active after
time $t$, i.e., $N(t) - t$ is geometrically distributed with parameter
$1/(1-\hat p)$.
\begin{align*}
  \E{S_t \given I_t = 1, Y_t = 0}
  &= \frac{\barv}{1 - \hat p q} - \frac{W}{1 - \hat p} + \Pr[I_{N(t)} = 1 \given
  I_t = 1] \E{S_{N(t)} \given I_{N(t)} = 1, Y_{N(t)} = 1}
\end{align*}
This is because in the time until system 2 is next active:
\begin{itemize}
  \item The agent gets expected utility $\barv$ as long as $I_{t'} = 1$ for each
    $t' \in \{t, \dots, N(t) - 1\}$. This is true with probability $q^{t'}$,
    leading to expected utility
    \begin{align*}
      \E{\sum_{t'=t}^{N(t) - 1} q^{t'-t} \barv} = \frac{\barv}{1 - \hat p q}
    \end{align*}
    because $N(t) - t$ is geometrically distributed, and given that system 1 is
    active, the agent chooses branch $i$ with probability proportional to $p_i$,
    leading to expected value
    \begin{align*}
      \frac{\sum_{i=1}^d p_i \barv_i}{\sum_{i=1}^d p_i} = \barv.
    \end{align*}
  \item The agent loses utility $W$ from forgoing its outside option for each
    $t' \in \{t, \dots, N(t) - 1\}$, leading to expected utility $-\frac{W}{1 -
    \hat p}$.
  \item At $N(t)$, system 2 is active. If $I_t = 1$, which happens with
    probability $\Pr[I_{N(t)} = 1 \given I_t = 1]$, the agent subsequently gets
    utility $\E{S_{N(t)} \given I_{N(t)} = 1, Y_{N(t)} = 1}$. Otherwise, the
    agent leaves the platform getting a net utility of 0.
\end{itemize}
Because the agent is memoryless, $\E{S_{N(t)} \given I_{N(t)} = 1, Y_{N(t)} = 1}
= \E{S_t \given I_t = 1, Y_t = 1}$. Putting this into~\eqref{eq:eS_breakdown},
\begin{align*}
  \E{S}
  &= \hat p \b{\frac{\barv}{1 - \hat p q} - \frac{W}{1 - \hat p} + \Pr[I_{N(t)} = 1
  \given I_t = 1] \E{S_{N(t)} \given I_{N(t)} = 1, Y_{N(t)} = 1}} \\
  &+ (1-\hat p) \E{S_t \given I_t = 1, Y_t = 1} \\
  &= \hat p \b{\frac{\barv}{1 - \hat p q} - \frac{W}{1 - \hat p}}
  + (1-\hat p + \hat p \Pr[I_{N(t)} = 1 \given I_t = 1]) \E{S_t \given I_t = 1,
  Y_t = 1}
  \numberthis \label{eq:eS_breakdown2}
\end{align*}
It therefore suffices to find $\E{S_t \given I_t = 1, Y_t = 1}$ and
$\Pr[I_{N(t)} = 1 \given I_t = 1]$. To do so, we begin by conditioning on $v_t$,
the highest value of any available option at time $t$.
\begin{equation}
  \E{S_t \given I_t = 1, Y_t = 1, v_t}
  = v_t + \frac{\vbar \hat p q}{1 - \hat p q} - \frac{W}{1 - \hat p}
  + \E{S_{N(t)} \given I_{N(t)} = 1} \Pr[I_{N(t)} = 1 \given I_t = 1]
  \label{eq:SIt}
\end{equation}
This is because
\begin{itemize}
  \item The $v_t$ term is the value of the content at this step.
  \item $\frac{\vbar \hat p q}{1 - \hat p q}$ gives the expected value of
    content chosen by system 1 until system 2 next wakes up.
  \item $\frac{W}{1 - \hat p}$ gives the expected value of the outside option
    until system 2 next wakes up.
  \item $\E{S_{N(t)} \given I_{N(t)} = 1} \Pr[I_{N(t)} = 1 \given I_t = 1]$
    gives the expected value starting from the next time system 2 is awake
    multiplied by the probability that the agent has not yet satiated.
\end{itemize}
Because the agent is memoryless,
\begin{align*}
  \E{S_{N(t)} \given I_{N(t)} = 1} = \E{S_t \given I_t = 1, Y_t = 1}.
\end{align*}
We can find $\psi \triangleq \Pr[I_{N(t)} = 1 \given I_t = 1]$ as follows.
\begin{align*}
  \psi
  &= \Pr[I_{N(t)} = 1 \given I_t = 1] \\
  &= q\b{(1 - \hat p) + \hat p \Pr[I_{N(t)} = 1 \given I_{t+1} = 1, N(t) > t+1]} \\
  &= q\b{(1 - \hat p) + \hat p \Pr[I_{N(t)} = 1 \given I_t = 1]} \tag{memoryless} \\
  &= q\b{(1 - \hat p) + \hat p \psi} \\
  \psi(1 - \hat p q)
  &= q (1 - \hat p) \\
  \psi
  &= \frac{q (1 - \hat p)}{1 - \hat p q}
\end{align*}
Plugging in to~\eqref{eq:SIt},
\begin{equation}
  \E{S_t \given I_t = 1, Y_t = 1, v_t}
  = v_t + \frac{\vbar \hat p q}{1 - \hat p q} - \frac{W}{1 - \hat p}
  + \E{S_t \given I_t = 1, Y_t = 1} \frac{q (1 - \hat p)}{1 - \hat p q}
  \label{eq:SIt2}
\end{equation}
Let $\mc V^d$ be the distribution of $v_t$, i.e., the max of $d$ independent
samples from $\mc V_1, \dots \mc V_d$. To get $\E{S_t \given I_t = 1}$, we have
\begin{align*}
  \E{S_t \given I_t = 1, Y_t = 1} = \int_{-\infty}^\infty \ind{\text{agent
  continues at $v_t$}} \cdot \E{S_t \given I_t = 1, Y_t = 1, v_t} \, d \mc V^d
\end{align*}
Of course, this is self-referential, since by~\eqref{eq:SIt2}, $\E{S_t \given
I_t = 1, Y_t = 1, v_t}$ depends on $\E{S_t \given I_t = 1, Y_t = 1}$; however,
note that $\E{S_t \given I_t = 1, Y_t = 1, v_t}$ is monotone increasing in
$v_t$, so the agent's policy must be to continue when $v_t$ is above some
threshold $\tau$:
\begin{align*}
  \E{S_t \given I_t = 1, Y_t = 1}
  &= \int_{\tau}^\infty \E{S_t \given I_t = 1, Y_t = 1, v_t} \, d \mc V^d \\
  &= \int_{\tau}^\infty \p{v_t + \frac{\vbar \hat p q}{1 - \hat p q} -
    \frac{W}{1 - \hat p} + \E{S_t \given I_t = 1, Y_t = 1} \frac{q (1 - \hat p)}{1
  - \hat p q}} \, d \mc V^d
\end{align*}
Define
\begin{align*}
  A(\tau)
  &\triangleq \int_{\tau}^\infty \p{v_t + \frac{\vbar \hat p q}{1 - \hat p q} -
  \frac{W}{1 - \hat p}} d \mc V^d
  = \E{v_t \cdot \ind{v_t \ge \tau}} + \p{\frac{\vbar \hat p q}{1 - \hat p q} -
  \frac{W}{1 - \hat p}} \Pr[v_t \ge \tau] \\
  B(\tau)
  &\triangleq \int_{\tau}^\infty \p{\frac{q (1 - \hat p)}{1 - \hat p q}} d \mc V
  = \p{\frac{q (1 - \hat p)}{1 - \hat p q}} \Pr[v_t \ge \tau] \\
  F(\tau) 
  &\triangleq A(\tau) + B(\tau) F(\tau) = \frac{A(\tau)}{1 - B(\tau)}
\end{align*}
By these definitions,
\begin{align*}
  \E{S_t \given I_t = 1, Y_t = 1}
  &= F(\tau^*)
\end{align*}
for the agent's threshold $\tau^*$. $F(\tau)$ gives the agent's expected utility
for using threshold $\tau$, so by definition, $\tau^*$ maximizes $F$. This
yields the agent's policy: the agent continues on the platform as long as $v_t
\ge \tau^*$, and our goal is now to find $\tau^*$ and $F(\tau^*)$. Define
\begin{align*}
  M(\gamma) \triangleq \max_\tau A(\tau) + B(\tau) \gamma.
\end{align*}
To find $M(\gamma)$, observe that
\begin{align*}
  \frac{d}{d\tau} A(\tau) + B(\tau) \gamma
  &= -\p{\tau + \frac{\vbar \hat p q}{1 - \hat p q} - \frac{W}{1 - \hat p} + \gamma
  \frac{q (1 - \hat p)}{1 - \hat p q}} f_{\mc V^d}(\tau)
\end{align*}
Choose $\tau_\gamma$ such that
\begin{align*}
  \tau_\gamma + \frac{\vbar \hat p q}{1 - \hat p q} - \frac{W}{1 - \hat p} + \gamma
  \frac{q (1 - \hat p)}{1 - \hat p q}
  &= 0 \\
  \tau_\gamma &\triangleq \frac{W}{1 - \hat p} - \frac{\vbar \hat p q}{1 - \hat
  p q} - \gamma \frac{q (1 - \hat p)}{1 - \hat p q}
\end{align*}
By this definition,
\begin{align*}
  \frac{d}{d\tau} A(\tau) + B(\tau) \gamma
  &\ge 0 && \forall{\tau < \tau_\gamma} \\
  \frac{d}{d\tau} A(\tau) + B(\tau) \gamma
  &\le 0 && \forall{\tau > \tau_\gamma}
\end{align*}
As a result, $\tau_\gamma$ weakly maximizes $A(\tau) + B(\tau) \gamma$, meaning
$M(\gamma) = A(\tau_\gamma) + B(\tau_\gamma) \gamma$. Thus, we can compute
$M(\gamma)$ for any $\gamma$ as long as we can compute $\E{v_t \cdot \ind{v_t
\ge \tau_k}}$ and $\Pr[v_t \ge \tau_k]$ for $v_t \sim \mc V^d$.

Using this, by Lemma~\ref{lem:gamma_simple}, we can find $\gamma^* \ge 0$
maximizing $M$ through binary search. Moreover, $\gamma^* = M(\gamma^*)$, i.e.,
$\gamma^*$ is a fixed point, and $\gamma^* = F(\tau^*) = \E{S_t \given I_t = 1,
Y_t = 1}$. We can find $\tau^* = \tau^*_{\gamma^*}$ as before, i.e,\footnote{In
general, $\tau^*$ will not be unique; however, it is \textit{functionally}
unique, in that if there is some other $\tau' \ne \tau^*$ that maximizes
$F(\cdot)$, then the distribution $\mc V$ has 0 probability mass between
$\tau^*$ and $\tau'$. Thus, the agent will behave identically under these
thresholds.}
\begin{align*}
  \tau^* &\triangleq \frac{W}{1 - \hat p} - \frac{\vbar \hat p q}{1 - \hat p q}
  - \gamma^* \frac{q (1 - \hat p)}{1 - \hat p q}.
\end{align*}
This fully specifies the agent's behavior: the agent leaves the platform when
system 2 is active and either $I_t = 0$ or $v_t < \tau^*$, and it remains on the
platform in all other cases. Note that if $\tau^*$ is sufficiently large (i.e.,
larger than any value in $\mc V$), the
agent will leave the platform whenever system 2 is active, regardless of $I_t$.
To find its expected utility, we can finally return to~\eqref{eq:eS_breakdown2}.
\begin{align*}
  \E{S}
  &= \hat p \b{\frac{\barv}{1 - \hat p q} - \frac{W}{1 - \hat p}}
  + (1-\hat p + \hat p \Pr[I_{N(t)} = 1 \given I_t = 1]) \E{S_t \given I_t = 1,
  Y_t = 1} \\
  &= \hat p \b{\frac{\barv}{1 - \hat p q} - \frac{W}{1 - \hat p}} + \p{1-\hat p
  + \hat p \frac{q (1 - \hat p)}{1 - \hat p q}} \E{S_t \given I_t = 1, Y_t = 1}
  \\
  &= \hat p \b{\frac{\barv}{1 - \hat p q} - \frac{W}{1 - \hat p}}
  + \p{1 - \frac{\hat p (1-q)}{1 - \hat p q}} \E{S_t \given I_t = 1, Y_t = 1} \\
  &= \hat p \b{\frac{\barv}{1 - \hat p q} - \frac{W}{1 - \hat p}}
  + \p{\frac{1 - \hat p}{1 - \hat p q}} \frac{\E{v_t \cdot
      \ind{v_t \ge \tau^*}} + \p{\frac{\vbar \hat p q}{1 - \hat p q} - \frac{W}{1
    - \hat p}} \Pr[v_t \ge \tau^*]}{1 - \p{\frac{q (1 - \hat p)}{1 - \hat p q}}
  \Pr[v_t \ge \tau^*]}
\end{align*}
If this is negative, the agent chooses not to use the platform at all, and
$\E{S} = 0$.

To find $\E{T}$, we could use a similar derivation to the one above, or we could
observe that we can find $\E{T}$ simply by setting $\vbar = 0$, $\E{v_t \cdot
\ind{v_t \ge \tau^*}} = 0$, and $W = -1$ in the expression for $\E{S}$, since
the agent loses utility $W$ to its outside option for every step it remains on
the platform.
\begin{align*}
  \E{T}
  &= \hat p \b{\frac{1}{1-\hat p}} + \p{\frac{1 - \hat p}{1 - \hat p q}}
  \frac{\p{\frac{1}{1-\hat p}} \Pr[v_t \ge \tau^*]}{1 - \p{\frac{q (1 - \hat
  p)}{1 - \hat p q}} \Pr[v_t \ge \tau^*]} \\
  &= \frac{\hat p}{1-\hat p} + \p{\frac{1}{1 - \hat p q}} \frac{\Pr[v_t \ge
  \tau^*]}{1 - \p{\frac{q (1 - \hat p)}{1 - \hat p q}} \Pr[v_t \ge \tau^*]}
\end{align*}
\end{proof}

\section{Supplementary lemmas}
\begin{lemma}
  \label{lem:half_lb}
  For $x > 0$,
  \begin{align*}
    \frac{x+1}{x} \p{1 - \frac{\ln (x+1)}{x}}
    &\ge \frac{1}{2}.
  \end{align*}
  For $-1 < x < 0$,
  \begin{align*}
    \frac{x+1}{x} \p{1 - \frac{\ln (x+1)}{x}}
    &\le \frac{1}{2}
  \end{align*}
\end{lemma}
\begin{proof}
  Since the function in question is continuously differentiable for $x \ne 0$,
  it suffices to show that
  \begin{enumerate}
    \item $\lim_{x \to 0} \frac{x+1}{x} \p{1 - \frac{\ln (x+1)}{x}} =
      \frac{1}{2}$
    \item For $x \in (-1, 0) \cap (0, \infty)$, $\frac{d}{dx} \frac{x+1}{x} \p{1
      - \frac{\ln (x+1)}{x}} \ge 0$.
  \end{enumerate}
  First, we have
  \begin{align*}
    \lim_{x \to 0} \frac{x+1}{x} \p{1 - \frac{\ln (x+1)}{x}}
    &= \lim_{x \to 0} \frac{x^2 + x - (x+1) \ln(x+1)}{x^2} \\
    &= \lim_{x \to 0} \frac{2x + 1 - 1 - \ln(x+1)}{2x} \\
    &= \lim_{x \to 0} \frac{2x - \ln(x+1)}{2x} \\
    &= \lim_{x \to 0} \frac{2 - \frac{1}{x+1}}{2} \\
    &= \frac{1}{2}
  \end{align*}
  Thus, it suffices to show that
  \begin{align*}
    \frac{d}{dx} \frac{x+1}{x} \p{1 - \frac{\ln (x+1)}{x}} \ge 0 ~~~~ \forall x \in
    (-1, 0) \cap (0, \infty) \\
  \end{align*}
  Taking the derivative,
  \begin{align*}
    \frac{d}{dx} \frac{x+1}{x} \p{1 - \frac{\ln (x+1)}{x}}
    &= -\frac{x+1}{x} \p{\frac{\frac{x}{x+1} - \ln(x+1)}{x^2}} - \p{1 -
    \frac{\ln (x+1)}{x}} \frac{1}{x^2} \\
    &= -\frac{1}{x^2} + \frac{(x+1)\ln(x+1)}{x^3} - \frac{1}{x^2} +
    \frac{\ln(x+1)}{x^3} \\
    &= \frac{(x+2) \ln (x+1) - 2x}{x^3}
  \end{align*}
  For $x > 0$,
  \begin{align*}
    \frac{(x+2) \ln (x+1) - 2x}{x^3} \ge 0
    \Longleftrightarrow (x+2) \ln (x+1) - 2x \ge 0,
  \end{align*}
  and for $x < 0$,
  \begin{align*}
    \frac{(x+2) \ln (x+1) - 2x}{x^3} \ge 0
    \Longleftrightarrow (x+2) \ln (x+1) - 2x \le 0.
  \end{align*}
  Thus, it suffices to show that $(x+2) \ln (x+1) - 2x$ is weakly positive for
  $x > 0$ and weakly negative for $x \in (-1, 0)$. To do so, we will show the
  following:
  \begin{enumerate}
    \item At $x = 0$, $(x+2) \ln(x+1) - 2x = 0$ (which holds simply by plugging
      in $x = 0$).
    \item $\frac{d}{dx} (x+2) \ln (x+1) - 2x \ge 0$ for $x > -1$.
  \end{enumerate}
  Taking the derivative,
  \begin{align*}
    \frac{d}{dx} (x+2) \ln (x+1) - 2x
    &= \frac{x+2}{x+1} + \ln(x+1) - 2 \\
    &= \frac{1}{x+1} + \ln(x+1) - 1
  \end{align*}
  Again, when $x = 0$,
  \begin{align*}
    \frac{1}{x+1} + \ln(x+1) - 1 = 0,
  \end{align*}
  so it suffices to show that this is a minimum for $x \in (-1, \infty)$, i.e.,
  \begin{align*}
    \frac{d}{dx} \frac{1}{x+1} + \ln(x+1) - 1
    \begin{cases}
      \ge 0 & \forall x > 0 \\
      \le 0 & \forall x \in (-1, 0)
    \end{cases}
  \end{align*}
  Taking the derivative,
  \begin{align*}
    \frac{d}{dx} \frac{1}{x+1} + \ln(x+1) - 1
    &= -\frac{1}{(x+1)^2} + \frac{1}{x+1} \\
    &= \frac{x}{(x+1)^2}
  \end{align*}
  This is positive for $x > 0$ and negative for $x \in (-1, 0)$ as desired.
  Thus, the lemma holds.
\end{proof}

\begin{lemma}
  \label{lem:gamma_simple}
  Let $\gamma^* = F(\tau^*)$.
  \begin{enumerate}
    \item For any $\gamma \in [0, \gamma^*]$, $M(\gamma) \ge \gamma$.
    \item For any $\gamma > \gamma^*$, $M(\gamma) < \gamma$.
  \end{enumerate}
  This implies $M(\gamma^*) = \gamma^* = F(\tau^*)$.
\end{lemma}
\begin{proof}
  For the first part, we'll show that in particular, $A(\tau^*) +
  B(\tau^*)\gamma \ge
  \gamma$. This is because
  \begin{align*}
    A(\tau^*) + B(\tau^*) \gamma
    &= A(\tau^*) + B(\tau^*) F(\tau^*) - B(\tau^*) F(\tau^*) + B(\tau^*) \gamma \\
    &= F(\tau^*) - B(\tau^*) F(\tau^*) + B(\tau^*) \gamma \\
    &= F(\tau^*)(1 - B(\tau^*)) + B(\tau^*) \gamma \\
    &\in [\gamma, F(\tau^*)]
  \end{align*}
  because $0 \le B(\cdot) < 1$.

  To show the second part, assume towards contradiction that there exists some
  $\tau$ and $\gamma > \gamma^*$ such that $A(\tau) + B(\tau) \gamma \ge
  \gamma$. As above,
  \begin{align*}
    A(\tau) + B(\tau) \gamma \in [\gamma, F(\tau)],
  \end{align*}
  so in order to have $A(\tau) + B(\tau) \gamma \ge \gamma$, we would need for
  this interval to lie above $\gamma$, i.e., $F(\tau) \ge \gamma > \gamma^* =
  F(\tau^*)$. But this is a contradiction because $F(\tau) \le F(\tau^*)$ for
  all $\tau$.

  Finally, observe that $M(\gamma^*) = \gamma^*$ because
  \begin{align*}
    \gamma^*
    &= F(\tau^*) \\
    &= \max_\tau A(\tau) + B(\tau) F(\tau^*) \\
    &= M(\gamma^*).
  \end{align*}
\end{proof}

\section{An Example of Tree Width Optimization}
\label{app:tree_ex}

\xhdr{Example \examp{ex:tree}: $d_S > d_T$}
Let $p = 0.01$, $q = 0$, and $W = 1$. The distribution $\mc V$ is
\begin{align*}
  v \sim \mc V
    &= \begin{cases}
      1.011 & \text{w.p. } 0.5 \\
      1.05 & \text{w.p. } 0.5
    \end{cases}
\end{align*}

In this example, the user wants to consume exactly one piece of content (since
$q = 0$). A higher $d$ means it will have better options from which to choose in
expectation, but it also means $\hat p$ will be higher. As a result, (1) it is
less likely that system 2 will be active at step $t = 0$, and (2) after the
user consumes the one piece of content it wants, it is more likely that system
1 will gain control and keep the user on the platform. We can use
Theorem~\ref{thm:tree_tau} to find $d_S = 2$ and $d_T = 1$. For $d = 1$, $\tau^*
= -\infty$, meaning the user will always consume content. For $d = 2$, however,
we can similarly derive $\tau^* = 1.05$, meaning if $v_0 = 1.011$ (the
highest-value content at step $t = 0$ has value $1.011$) and system 2 is active
at $t = 0$ (which happens with probability $1-\hat p$), the user will leave
without consuming anything. The user gets higher utility at $d = 2$ than $d =
1$ because it is more likely to see high-value content; however, it consumes
less because with some probability, it encounters only low-value content at $t =
0$ and immediately leaves without consuming anything.

\end{document}